\newif\iflong
\longtrue
\longfalse

\documentclass[letterpaper,11pt]{article}
\usepackage[margin=1in]{geometry}
\usepackage{amsfonts,amsmath,amssymb,amsthm}
\usepackage{thmtools} \usepackage{thm-restate}
\usepackage[noadjust]{cite}
\usepackage{comment,enumerate,hyperref}
\usepackage[colorinlistoftodos,textsize=small,color=red!25!white,obeyFinal]{todonotes}
\newtheorem{theorem}{Theorem}[section]

\newtheorem{lemma}[theorem]{Lemma}

\usepackage{soul}
\theoremstyle{definition}
\newtheorem{definition}{Definition}

\usepackage{xcolor,xspace}
\usepackage{booktabs}

\usepackage[boxruled, noend]{algorithm2e}
\usepackage{placeins}

\usepackage{csquotes}

\newcommand{\Span}{\text{Span}}

\DeclareMathOperator*{\argmax}{arg\,max}

\title{Relative Survivable Network Design}

 \author{Michael Dinitz\thanks{Supported in part by NSF award CCF-1909111.} \\ Johns Hopkins University\\ \texttt{mdinitz@cs.jhu.edu}\and
Ama Koranteng\footnotemark[1] \\ Johns Hopkins University\\ \texttt{akorant1@jhu.edu} \and Guy Kortsarz\\Rutgers University, Camden\\ \texttt{guyk@camden.rutgers.edu}}

\date{}

\begin{document}
\begin{titlepage}
\maketitle

\begin{abstract} 
    One of the most important and well-studied settings for network design is edge-connectivity requirements.  This encompasses uniform demands such as the Minimum $k$-Edge-Connected Spanning Subgraph problem ($k$-ECSS), as well as nonuniform demands such as the Survivable Network Design problem.  A weakness of these formulations, though, is that we are not able to ask for fault-tolerance larger than the connectivity.  Taking inspiration from recent definitions and progress in graph spanners, we introduce and study new variants of these problems under a notion of \emph{relative} fault-tolerance.  Informally, we require not that two nodes are connected if there are a bounded number of faults (as in the classical setting), but that two nodes are connected if there are a bounded number of faults \emph{and the two nodes are connected in the underlying graph post-faults}.  That is, the subgraph we build must ``behave'' identically to the underlying graph with respect to connectivity after bounded faults. 
    
    We define and introduce these problems, and provide the first approximation algorithms: a $(1+4/k)$-approximation for the unweighted relative version of $k$-ECSS, a $2$-approximation for the weighted relative version of $k$-ECSS, and a $27/4$-approximation for the special case of Relative Survivable Network Design with only a single demand with a connectivity requirement of $3$.  To obtain these results, we introduce a number of technical ideas that may of independent interest.  First, we give a generalization of Jain's iterative rounding analysis that works even when the cut-requirement function is not weakly supermodular, but instead satisfies a weaker definition we introduce and term \emph{local} weak supermodularity.  Second, we prove a structure theorem and design an approximation algorithm utilizing a new decomposition based on \emph{important separators}, which are structures commonly used in fixed-parameter algorithms that have not commonly been used in approximation algorithms.
\end{abstract}
\thispagestyle{empty}
\end{titlepage}

\renewcommand{\arraystretch}{1.5} 

\section{Introduction}
Fault-tolerance has been a central object of study in approximation algorithms, particularly for network design problems where the graphs that we study represent some physical objects which might fail (communication links, transportation links, etc.).  In these settings it is natural to ask for whatever object we build to be fault-tolerant.  The precise definition of ``fault-tolerance" is different in different settings, but a common formulation is edge fault-tolerance, which typically takes the form of edge connectivity.  Informally, these look like guarantees of the form ``if up to $k$ edges fail, then the nodes I want to be connected are still connected.''  For example, consider the following two classical problems.
\begin{itemize}
    \item The Minimum $k$-Edge Connected Subgraph problem ($k$-ECSS), where we are given a graph $G$ and a value $k$ and are asked to find the $k$-edge connected subgraph of $G$ of minimum size (or weight).  In other words, if fewer than $k$ edges fail, the graph should still be connected.
    \item The more general Survivable Network Design problem (SND, sometimes referred to as Generalized Steiner Network), where we are given a graph $G$ and demands $\{(s_i, t_i, k_i)\}_{i \in [\ell]}$, and we are supposed to find the minimum-weight subgraph $H$ of $G$ so that there are at least $k_i$ edge-disjoint paths between $s_i$ and $t_i$ for every $i \in [\ell]$.  In other words, for every $i \in [\ell]$, if fewer than $k_i$ edges fail then $s_i$ and $t_i$ will still be connected in $H$ even after failures.
\end{itemize}

Both of these problems have been studied extensively (for a small sample, see~\cite{WGMV95,Jain01,CT00,GGTW09}), and are paradigmatic examples of network design problems.  But there is a different notion of fault-tolerance which is stronger, and in some ways more natural: \emph{relative} fault-tolerance.  Relative fault-tolerance makes guarantees that rather than being absolute (``if at most $k$ edges fail the network still functions") are relative to an underlying graph or system (``if at most $k$ edges fail, the subgraph functions just as well as the original graph post-failures").  This allows us to generalize the traditional definition: if the underlying graph has strong enough connectivity properties then the two definitions are the same, but the relative version allows us to make interesting and nontrivial guarantees even when the underlying graph does not have strong connectivity properties.  

For example, the definition of Survivable Network Design has an important limitation: if $G$ itself can only support a small number of edge disjoint $s_i - t_i$ paths (e.g., $3$), then of course we cannot ask for a subgraph with more edge-disjoint paths.  There simply would be no feasible solution.  But this is somewhat unsatisfactory.  For example, while we cannot guarantee that $s_i$ and $t_i$ would be connected after \emph{any} set of $5$ faults (since those faults may include an $(s_i, t_i)$ cut of size $3$), clearly there could be \emph{some} set of $5$ faults  which do not in fact disconnect $s_i$ from $t_i$ in $G$.  And if these faults occur, it is natural to want $s_i$ and $t_i$ to still be connected in (what remains) of $H$.  In other words: just because there exists a small cut, why should we give up on being tolerant to a larger number of faults which do not contain that cut?

\subsection{Our Results and Techniques}
In this paper we initiate the study of relative fault-tolerance in network design, by defining relative versions of Survivable Network Design and $k$-ECSS.

\begin{definition} \label{def:RSND}
In the Relative Survivable Network Design problem (RSND), we are given a graph $G = (V, E)$ with edge weights $w : E \rightarrow \mathbb{R}_{\geq 0}$ and demands $\{(s_i, t_i, k_i)\}_{i \in [\ell]}$.  A feasible solution is a subgraph $H$ of $G$ where for all $i \in [\ell]$ and $F \subseteq E$ with $|F| < k_i$, if there is a path in $G \setminus F$ from $s_i$ to $t_i$ then there is also a path in $H \setminus F$ from $s_i$ to $t_i$.  Our goal is to find the minimum weight feasible solution.  
\end{definition}

\begin{definition} \label{def:kEFTS}
The $k$-Edge Fault-Tolerant Subgraph problem ($k$-EFTS) is the special case of RSND where there is a demand between all pairs and every $k_i$ is equal to $k$.  In other words, we are given a graph $G = (V, E)$ with edge weights $w : E \rightarrow \mathbb{R}_{\geq 0}$.  A feasible solution is a subgraph $H$ of $G$ where for all $F \subseteq E$ with $|F| < k$, any two nodes which have a path between them in $G \setminus F$ also have a path between them in $H \setminus F$ (the connected components of $H \setminus F$ are identical to the connected components of $G \setminus F$).  Our goal is to find the minimum weight feasible solution. 
\end{definition}

For both of these problems, we say that they are \emph{unweighted} if all edges have the same weight (or equivalently $w(e) = 1$ for all $e \in E$).  Note that if $s_i$ and $t_i$ are $k_i$-connected in $G$ for every $i \in [\ell]$, then RSND is exactly the same as SND, and if $G$ is $k$-connected then $k$-EFTS is exactly the same as $k$-ECSS.  Hence we have strictly generalized these classical problems.

We note that the fault-tolerance we achieve is really ``one less'' than the given number (there are strict inequalities in the definitions).  This is ``off-by-one'' from the related relative fault-tolerance literature~\cite{CLPR10,BDR21,BDR22}, but makes the connection to SND and $k$-ECSS cleaner.  

\paragraph{Difficulties.} Before discussing our results or techniques, we briefly discuss what makes these problems difficult.  The non-relative versions are classical and have been studied extensively: why can't we just re-use the ideas and techniques developed for them?  Particularly since there is only a difference in the setting when there are small cuts in the graph, in which case we already know that the edges of those cuts must be included in any feasible solution?  

Unfortunately, it turns out that this seemingly minor change has a dramatic impact on the structure of the problem.  Most importantly, the \emph{cut requirement function} has dramatically different properties.  In $k$-ECSS, Menger's theorem implies that $H$ is a valid solution if and only if for all $S \subset V$ with $S \neq \emptyset$, there are at least $k$ edges between $S$ and $\bar S$.  Hence we can rephrase $k$-ECSS as the problem of finding a minimum cost subgraph such that that there are at least $f(S)$ edges across the cut $(S, \bar S)$ for all $S \subset V$ with $S \neq \emptyset$, where $f(S) = k$.  Similarly, we can rephrase SND as the same problem but where $f(S) = \max_{i \in[\ell] : s_i \in S, t_i \not\in S} k_{i}$ (as was shown in~\cite{WGMV95}).  Thus both problems can be thought of as choosing a minimum cost subgraph subject to satisfying some cut-requirement covering function $f : 2^V \rightarrow \mathbb{R}$.  So a natural starting point for any approximation algorithm is to write the natural covering LP relaxation which has a covering constraint of $f(S)$ for every cut $S$.  And indeed, the covering LP using the cut-requirement function was the starting point for both the primal-dual $O(\max_{i \in \ell} k_i)$-approximation for SND of~\cite{WGMV95} and the seminal $2$-approximation for SND using iterative rounding due to Jain~\cite{Jain01}.  It has also been used for $k$-ECSS~\cite{GGTW09}, although (unlike SND) there are also purely combinatorial approximations~\cite{CT00}.

Hence the natural starting point for us to study RSND and $k$-EFTS would be to formulate them in terms of cut-requirement functions and try the same approaches as were used in SND and $k$-ECSS.  But this is easier said than done.  The functions are a little more complicated, but it is not too hard to construct a cut requirement function that characterizes feasible solutions.  However, in order to use the iterative rounding technique of Jain~\cite{Jain01} (or any of the weaker techniques which it superceded), the cut requirement function needs to have a structural property known as \emph{weak} (or \emph{skew}) \emph{supermodularity}~\cite{Jain01}.  This turns out to be crucial, and there are still (to the best of our knowledge) no successful uses of iterative rounding in settings without weak supermodularity. And unfortunately, it turns out that our cut requirement functions \emph{are not} weakly supermodular.  So while we can phrase our problems as satisfying a cut requirement function, we cannot actually use iterative rounding, uncrossing, or any other part of the extensive toolkit that has grown around~\cite{Jain01}.  

\paragraph{Our approaches.} We get around this difficulty in two ways.  For $k$-EFTS, we define a new property of cut requirement functions which we call \emph{local} weak supermodularity, and prove that our cut requirement function has this property and that it is sufficient for iterative rounding.  This is, to the best of our knowledge, \emph{the first} use of iterative rounding without weak supermodularity.  For RSND with a single demand, we use an entirely different combinatorial approach based on decomposing the graph into a chain of connected components using \emph{important separators}~\cite{M11}, an important tool from fixed-parameter tractability that, to the best of our knowledge, has not been used before in approximation algorithms.  

\subsubsection{$k$-Edge Fault-Tolerant Subgraph} \label{intro:kEFTS}
We begin in Section~\ref{sec:kEFTS} with $k$-EFTS, where we prove the following two theorems.

\begin{theorem} \label{thm:kEFTS-weighted}
There is a polynomial-time $2$-approximation for the $k$-EFTS problem.
\end{theorem}

\begin{theorem} \label{thm:kEFTS-unweighted}
There is a polynomial-time $(1+4/k)$-approximation for the unweighted $k$-EFTS problem.
\end{theorem}

Both of these theorems are consequences of a structural property we prove about the cut-requirement function for $k$-EFTS: while it is not weakly supermodular, it does have a weaker property which we term \emph{local} weak supermodularity.  We define this property formally in Section~\ref{sec:LWS}, but at a high level it boils down to proving that while the inequalities required for weak supermodularity do not hold everywhere (as would be required for weak supermodularity), they hold for \emph{particular} sets (i.e., they hold \emph{locally}) which are the sets where the inequalities are actually applied by Jain's analysis.  In other words, we prove that the places in the function where weak supermodularity are violated are precisely the places where we do not care if weak supermodularity holds.  After overcoming a few more technical complications (we actually need local weak supermodularity even in the ``residual'' problem to use iterative rounding), this means that we can apply Jain's algorithm to prove Theorem~\ref{thm:kEFTS-weighted}.  

To prove Theorem~\ref{thm:kEFTS-unweighted}, it was observed for unweighted $k$-ECSS by~\cite{GGTW09} (with later improvements by~\cite{GG12}) that one of the main pieces of Jain's approach, the fact that the tight constraints can be ``uncrossed" to get a laminar family with the same span, implies a $(1+4/k)$-approximation via a trivial threshold rounding.  They pointed out that the fact that the linearly independent tight constraints form a laminar family implies that there are only $2n$ linearly independent tight constraints, while there are $m$ variables, and hence at any basic feasible solution the remaining $m-2n$ tight constraints defining the point must be the bounding constraints.  These bounding constraints being tight means that the associated variables are in $\{0,1\}$, and hence there are only $2n$ fractional variables in any basic feasible solution.  Rounding all of these variables to $1$ increases the cost by $2n$, but since $OPT \geq kn/2$ (since the input graph $G$ must be $k$-connected) this results in a $(1+4/k)$-approximation.

Thanks to our local weak supermodularity characterization the laminar family result is still true even for $k$-EFTS, so it is still true that there are at most $2n$ nonzero variables at any extreme point.  But since we are not guaranteed that $G$ is $k$-connected we are not guaranteed that $OPT \geq kn/2$, and so this does not imply the desired approximation.  Instead, we prove that the number of fractional variables at any basic feasible solution is at most $2 n_h$, where $n_h$ is the number of ``high-degree'' nodes.  It is then easy to argue that $OPT \geq n_h k / 2$, which gives Theorem~\ref{thm:kEFTS-unweighted}.

\subsubsection{Relative Survivable Network Design With a Single Demand} \label{intro:RSND}

For $k$-EFTS, we strongly used the property that all pairs have the same demand.  This is not true for RSND, which makes the problem vastly more difficult.  We still do not know whether there exists a cut requirement function which characterizes the problem and is locally weakly supermodular.  In this paper, we study the simplest case where not all demands are the same: when there is a single nonzero demand $(s,t,k)$, and $k$ is either $2$ or $3$ (the case of $k=1$ is simply the shortest-path problem). It turns out to be relatively straightforward to prove a $2$-approximation for $k=2$ even when there are many demands (see Section~\ref{sec:reduction}), but the $k=3$ case is surprisingly difficult.  We prove the following theorem in Section~\ref{sec:RSND}.

\begin{theorem} \label{thm:3RSND}
In any RSND instance with a single demand $(s,t,3)$, there is a polynomial-time $7-\frac14 = \frac{27}{4}$-approximation.
\end{theorem}

To prove this, we start with the observation that if the minimum $s-t$ cut is at least $3$ then this is actually just the traditional SND problem (and in fact, the even simpler problem of finding $3$ edge-disjoint paths of minimum weight between $s$ and $t$, which can be solved efficiently via min-cost flow).  So the only difficulty is when there are cuts of size $1$ or $2$.  Cuts of size $1$ can be dealt with easily (see Section~\ref{sec:reduction}), but cuts of size $2$ are more difficult.  To get rid of them, we construct a ``chain'' of $2$-separators (cuts of size $2$ that are also \emph{important separators}~\cite{M11}).  Inside each component of the chain there are no $2$-cuts between the incoming separator and the outgoing separator, which allows us to characterize the connectivity requirement of any feasible solution restricted to that component.  These connectivity requirements turn out to be quite complex even though we started with only a single demand, as fault sets with different structure can force complicated connectivity requirements in intermediate components.  The vast majority of the technical work is proving a structure lemma which characterizes them.  With this lemma in hand, though, we can simply approximate the optimal solution in each component.

Interestingly, to the best of our knowledge this is the first use of important separators in approximation algorithms, despite their usefulness in fixed-parameter algorithms~\cite{M11}.  

\subsection{Related Work}

The most directly related work is the $2$-approximation of Jain for Survivable Network Design~\cite{Jain01}, which introduced iterative rounding (see~\cite{LRS11} for a detailed treatment of iterative rounding in combinatorial optimization).  This built off of an earlier line of work on survivable network design beginning over 50 years ago with~\cite{SWK69}.  Since the success of Jain's approach for SND, there has been a significant amount of work on vertex-connectivity versions rather than edge-connectivity, which is a significantly more difficult setting.  This has culminated in the state of the art approximation of~\cite{CK12}.  There is also a long line of work on $k$-ECSS, most notably including~\cite{CT00,GGTW09}.

While not technically related, the basic problems in this paper are heavily inspired by recent work on relative notions of fault-tolerance in graph spanners and other non-optimization network design settings.  A relative definition of fault-tolerance for graph spanners which is very similar to ours (but which takes distances into account due to the spanner setting) was introduced by~\cite{CLPR10}, who gave bounds on the size of $f$-fault-tolerant $t$-spanners for both edge and vertex notions of fault-tolerance.  This spawned a line of work which improved these bounds for both vertex and edge fault-tolerance~\cite{DK11,BDPW18,BP19,DR20,BDR21,BDR22}, culminating in~\cite{BDR21} for vertex faults and~\cite{BDR22} for edge faults.  The basic spanner definition also inspired work on relative fault-tolerant versions of related problems, including emulators~\cite{BDN22}, distance sensitivity oracles for multiple faults~\cite{CLPRoracles}, and single-source reachability subgraphs~\cite{BCR16,LMSZ19}.  What all of these results shared, though, was that they were not doing optimization: they were looking for existential bounds (and algorithms to achieve them) for these objects.  In this paper, by contrast, we take the point of view of optimization and approximation algorithms and compare to the instance-specific optimal solution.

\section{$k$-Edge Fault-Tolerant Subgraph} \label{sec:kEFTS}
Both Theorems~\ref{thm:kEFTS-weighted} and~\ref{thm:kEFTS-unweighted} depend on the same LP relaxation, which is based on a modification of the ``obvious'' cut-requirement function.  So we begin by discussing this relaxation, and then use it to prove the two main theorems.  \iflong \else All missing proofs can be found in Appendix~\ref{app:kEFTS}. \fi

\subsection{LP Relaxation}
\subsubsection{Basics}
The natural place to start is the LP used by Jain~\cite{Jain01}, but with a cut requirement function $f(S) = \min(|\delta_G(S)|, k)$.  Unfortunately, while this results in a valid LP relaxation, it is not weakly supermodular (see Section~\ref{sec:LWS} for the definition, and Appendix~\ref{app:counterexamples} for a counterexample).  So instead we modify this cut requirement function by removing edges which are ``forced''.  For every subset $S$ of $V$, let $\delta_{G}(S)$ be the set of edges with exactly one endpoint in $S$.  Let $F = \{e \in E \mid \exists S \text{ where } e \in \delta_{G}(S) \text{ and } |\delta_{G}(S)| \leq k  \}$. In other words, $F$ is the set of all edges that are in some cut of size at most $k$.  Clearly we can compute $F$ in polynomial time by simply checking for every edge $(u,v)$ whether the minimum $u-v$ cut in $G$ has size at most $k$. For every set $S \subset V$ with $S \neq\emptyset$, we define the cut requirement function $f_F(S) = \min(k, |\delta_{G}(S)|) - |\delta_{G}(S) \cap F|$.  Note that every edge in $F$ must be in any feasible solution, since if any edge is missing then a fault set consisting of the rest of the cut (at most $k-1$ edges) would disconnect the endpoints of the missing edge in the solution but not in $G$, giving a contradiction.  Then $f_F(S)$ is essentially the ``remaining requirement'' after $F$ has been removed.  

Since iterative rounding will add other edges and remove them from the residual problem, we will want to define a similar cut requirement function for supersets: formally, for any $F' \supseteq F$, let $f_{F'}(S) = \min(k, |\delta_{G}(S)|) - |\delta_{G}(S) \cap F'|$.  For any $F' \supseteq F$, consider the following linear program which we call LP($F'$), which has a variable $x_e$ for every edge $e \in E \setminus F'$:

\begin{equation} \label{LP:main} \tag{LP($F'$)}
\fbox{$
\begin{array}{lll}
\min & \displaystyle \sum_{e \in E \setminus F'} w(e) x_e \\
\mathrm{s.t.} & \displaystyle \sum_{e \in \delta_G(S) \setminus F'} x_e \geq f_{F'}(S) \qquad & \forall S \subseteq V \\
& \displaystyle 0 \leq x_e \leq 1 & \forall e \in E \setminus F'
\end{array}
$
}
\end{equation}

It is not hard to see that this is a valid LP relaxation (when combined with $F'$), but we prove this for completeness.  

\begin{lemma} \label{lem:relaxation}
Let $H$ be a valid $k$-EFTS and let $F' \supseteq F$.  For every edge $e \in E \setminus F'$, let $x_e = 1$ if $e \in H$, and let $x_e = 0$ otherwise.  Then $x$ is a feasible integral solution to LP($F'$).
\end{lemma}

\begin{lemma} \label{lem:relaxation2}
Let $F' \supseteq F$ and let $x$ be an integral solution to LP($F'$).  Let $E' = \{e : x_e = 1\}$.  Then $H = E' \cup F'$ is a valid $k$-EFTS.
\end{lemma}

These lemmas (together with the fact that every edge in $F$ must be in any valid solution) imply that if we can solve and round this LP while losing some factor $\alpha$, then we can add $F$ to the rounded solution to get an $\alpha$-approximation.  Hence we are interested in solving and rounding this LP.

We first argue that we can solve the LP using the Ellipsoid algorithm with a separation oracle.  Note that unlike $k$-ECSS, here a violated constraint does not just correspond to a cut with LP values less than $k$, since our cut-requirement function is more complicated.  Indeed, if we compute a global minimum cut (with respect to the LP values) then we may end up with a small cut which is not violated even though there are violated constraints.  So we need to argue more carefully that we can find a violated cut when one exists.

\begin{lemma} \label{lem:solve}
For every $F' \supseteq F$, LP($F'$) can be solved in polynomial time.
\end{lemma}

\subsubsection{Local Weak Supermodularity} \label{sec:LWS}

As discussed in Section~\ref{intro:kEFTS}, it would be nice if this LP were \emph{weakly supermodular}, as this would immediately let us apply Jain's iterative rounding algorithm to obtain a $2$-approximation.  Recall the definition of weak supermodularity from~\cite{Jain01}.

\begin{definition}
Let $f : 2^V \rightarrow \mathbb{Z}$.  Then $f$ is \emph{weakly supermodular} if for every $A, B \subseteq V$, either $f(A) + f(B) \leq f(A \setminus B) + f(B \setminus A)$, or $f(A) + f(B) \leq f(A \cap B) + f(A \cup B)$.
\end{definition}

Unfortunately, our cut requirement function is not weakly supermodular; see Appendix~\ref{app:counterexamples} for a counterexample.  But we can make a simple observation that, to the best of our knowledge, has not previously been noticed or utilized in iterative rounding: Jain's iterative rounding algorithm does not actually need the weak supermodularity conditions to hold for \emph{all} pairs of sets $A, B$.  It only needs weak supermodularity to ``uncross'' the tight sets of an LP solution into a laminar family of tight sets with the same span.  Recall that a set is tight in a given LP solution if its corresponding cut constraint is tight, i.e., is satisfied with equality.  Moreover, note that in our setting, depending on our choice of $F'$ some cuts might be entirely included in $F'$.  These cuts would not have any edges remaining, resulting in an ``empty'' constraint in LP($F'$).  Such a constraint cannot be tight by definition, and also is not linearly independent with any other set of constraints.

Hence in order to use Jain's iterative rounding, we simply need our cut-requirement function $f_{F'}$ to satisfy the weak supermodularity requirements for $A,B$ where there is actually a nontrivial constraint for $A, B$ and where $F' \supseteq F$ (here $F'$ will consist of $F$ together with edges that Jain's iterative rounding algorithm has already set to $1$).  We formalize this as follows.  Given $F' \supseteq F$, we say that $S$ is an \emph{empty cut} if $\delta_G(S) \cap F' = \delta_G(S)$, and otherwise it is \emph{nonempty}. 

\begin{definition}
Given a graph $G = (V, E)$, a set $F' \subseteq E$, and a function $g : 2^V \rightarrow \mathbb{Z}$, we say that $g$ is \emph{locally weakly supermodular} with respect to $F'$ if for every $A, B \subseteq V$ with both $A$ and $B$ nonempty cuts, at least one of the following conditions holds:
\begin{itemize}
    \item $g(A) + g(B) \leq g(A \setminus B) + g(B \setminus A)$, or
    \item $g(A) + g(B) \leq g(A \cap B) + g(A \cup B)$.
\end{itemize}
\end{definition}

We will now prove that for any $F' \supseteq F$, the function $f_{F'}$ is locally weakly supermodular with respect to any $F'$.  
\iflong
\else
This is the key technical idea enabling Theorems~\ref{thm:kEFTS-weighted} and \ref{thm:kEFTS-unweighted}.
\fi 

\begin{theorem}[Local Weak Supermodularity] \label{thm:LWS}
For any $F' \supseteq F$, the cut requirement function $f_{F'}$ is locally weakly supermodular with respect to $F'$. 
\end{theorem}

\subsection{Unweighted $k$-EFTS}
To prove Theorem~\ref{thm:kEFTS-unweighted} we need to look inside~\cite{Jain01}.  The following two lemmas from~\cite{Jain01} are the main ``uncrossing'' lemmas which depend on weak supermodularity, and in which we can use local weak supermodularity instead without change.  As in~\cite{Jain01}, for each $S \subseteq V$ we use $\mathcal A_G(S)$ to denote the row of the constraint matrix corresponding to $S$.  In other words $\mathcal A_G(S)$ is a vector indexed by elements of $E \setminus F$ which has a $1$ in the entry for $e$ if $e \in \delta_G(S) \setminus F$, and otherwise has a $0$ in that entry.

\begin{lemma}[Lemma 4.1 of~\cite{Jain01}] \label{lem:uncrossing}
If two sets $A$ and $B$ are tight then at least one of the following must hold
\begin{enumerate}
    \item $A \setminus B$ and $B \setminus A$ are also tight, and $\mathcal A_G(A) + \mathcal A_G(B) = \mathcal A_G(A \setminus B) + \mathcal A_G(B \setminus A)$
    \item $A \cap B$ and $A \cup B$ are also tight, and $\mathcal A_G(A) + \mathcal A_G(B) = \mathcal A_G(A \cap B) + \mathcal A_G(A \cup B)$
\end{enumerate}
\end{lemma}

Let $\mathcal T$ denote the family of all tight sets.  For any family $\mathcal F$ of tight sets, let $\Span(\mathcal F)$ denote the vector space spanned by $\{\mathcal A_G(S) : S \in \mathcal F\}$.

\begin{lemma}[Lemma 4.2 of~\cite{Jain01}] \label{lem:laminar-span}
For any maximal laminar family $\mathcal L$ of tight sets, $\Span(\mathcal L) = \Span(\mathcal T)$.
\end{lemma}

Recall that $n_h$ is the number of high-degree nodes, i.e., nodes of degree at least $k$ in $G$.  Then we have the following lemma, which is a modification of Lemma 4.3 of~\cite{Jain01} where we give a stronger bound on the number of sets.

\begin{lemma} \label{lem:tight-dimension}
The dimension of $\Span(\mathcal T)$ is at most $2n_h - 1$.
\end{lemma}

We can now prove Theorem~\ref{thm:kEFTS-unweighted}.

\begin{proof}[Proof of Theorem~\ref{thm:kEFTS-unweighted}]
We first solve LP($F$) using Lemma~\ref{lem:solve} to get some basic feasible solution $x$.  Since there are $|E \setminus F|$ variables, this point is defined by $|E \setminus F|$ linearly independent tight constraints.  Lemma~\ref{lem:tight-dimension} implies that at most $2n_h-1$ of these are from tight sets, and hence all of the other tight constraints must be of the form $x_e = 0$ or $x_e = 1$ for some edge $e \in E \setminus F$.  Thus at most $2n_h - 1$ edges are assigned a fractional value in $x$.  Hence if we include all such edges in our solution $H$, together with all edges with $x_e = 1$ and all edges in $F$, we have a solution which is feasible (by Lemma~\ref{lem:relaxation2}).  Note that any high-degree node must have degree at least $k$ in any feasible solution, and thus $OPT \geq \frac{k}{2} n_h$.  Hence our solution $H$ has size at most
\begin{align*}
    |H| &\leq \sum_{e \in E \setminus F} x_e + |F| + 2n_h \leq OPT + 2n_h\leq OPT + \frac{4}{k} OPT = \left(1+\frac{4}{k}\right)OPT. \qedhere
\end{align*}
\end{proof}

\subsection{Weighted $k$-EFTS}
Jain's approximation algorithm solves the initial LP, rounds up and removes any edges with $x_e \geq 1/2$ which results in a residual problem, and repeats.  This is obviously a $2$-approximation (see~\cite{Jain01} for details), but requires proving that there is always at least one edge with $x_e \geq 1/2$ so we can make progress (even in the residual problems).  This is accomplished by proving Lemmas~\ref{lem:uncrossing} and~\ref{lem:laminar-span} to show that the tight constraints can be ``uncrossed'' into a laminar family.  This requires weak supermodularity, but as discussed, since in our LP every tight constraint must be a nonempty constraint, it is sufficient to replace this with local weak supermodularity.  Jain then uses a complex counting argument based on this laminar family of tight constraints to prove that some edge $e$ must have $x_e \geq 1/2$.  Importantly, nothing in this counting argument depends on the cut requirement having any particular structure (e.g., weak supermodularity); it depends only on the fact that the family of tight constraints can be uncrossed to be laminar.  

Since local weak supermodularity is sufficient to uncross the tight constraints into a laminar family, we can simply apply Jain's counting argument on this family for LP($F'$) to obtain the following lemma (as in Theorem 3.1 of~\cite{Jain01}).
\begin{lemma} \label{lem:semi-integral}
For all $F' \supseteq F$, in any basic feasible solution $x$ of LP($F'$) there is at least one $e \in E \setminus F'$ with $x_e \geq 1/2$.
\end{lemma}

\iflong
\else
When combined with the rest of the analysis in~\cite{Jain01} (particularly Theorem 3.2), this implies that iterative rounding is a $2$-approximation, implying Theorem~\ref{thm:kEFTS-weighted}.
\fi

\iflong
\else
\section{2-Connectivity and $k=2$} \label{sec:reduction}
We will now move on from $k$-EFTS to the more general RSND problem.  It turns out to be relatively straightforward to handle cuts of size $1$: removing such cuts gives a tree of $2$-connected components, and we can essentially run an algorithm independently inside each component.  This gives the following theorem, the proof of which can be found in Appendix~\ref{app:reduction}.  
\begin{theorem}
\label{thm:2connected}
If there exists an $\alpha$-approximation algorithm for RSND on 2-edge connected graphs, then there is an $\alpha$-approximation algorithm for RSND on general graphs. 
\end{theorem}

Extending this slightly gives the following theorem (proof in Appendix~\ref{app:2-RSND}), where $2$-RSND denotes the special case of the RSND problem where $k_i \leq 2$ for all $i$.

\begin{theorem} \label{thm:2RSND}
There is a $2$-approximation algorithm for $2$-RSND.
\end{theorem}
\fi

\section{RSND with a Single Demand: $k=3$} \label{sec:RSND}
In this section we prove Theorem~\ref{thm:3RSND}.  In the Single Demand RSND problem, we are given a graph $G=(V,E)$ (possibly with edge weights $w : E \rightarrow \mathbb{R}^+$) and a $k$-relative fault tolerance demand for a single vertex pair $(s,t)$. In other words, the set of connectivity demands is just $\{(s,t,k)\}$.
We give a $7-\frac14 = \frac{27}{4}$-approximation algorithm for the $k=3$ Single Demand RSND problem. The main idea is to partition the input graph using important separators, prove a structure lemma which characterizes the required connectivity guarantees within each component of the partition, and then achieve these guarantees using a variety of subroutines: a min-cost flow algorithm, a 2-RSND approximation algorithm (Theorem~\ref{thm:2RSND}), and a Steiner Forest approximation algorithm~\cite{AKR95}.  

\subsection{Decomposition} \label{sec:decomposition}
By Theorem~\ref{thm:2connected}, an $\alpha$-approximation algorithm for RSND on $2$-connected graphs implies an $\alpha$-approximation algorithm for RSND on general graphs. Hence going forward, we will assume the input graph $G$ is 2-connected. In this section we define important separators and describe how to construct what we call the \emph{$s-t$ 2-chain} of $G$. 

\begin{definition}
Let $X$ and $Y$ be vertex sets of a graph $G$. An \textit{$(X,Y)$-separator} of $G$ is a set of edges $S$ such that there is no path between any vertex $x \in X$ and any vertex $y \in Y$ in
$G \setminus S$. An $(X,Y)$-separator $S$ is \textit{minimal} if no subset $S' \subset S$ is also an $(X,Y)$-separator.  If $X = \{x\}$ and $Y = \{y\}$, we say that $S$ is an $(x,y)$-separator.
\end{definition}

The next definition, which is a slight modification of the definition due to~\cite{M11}, is a formalization of a notion of a ``closest'' separator.  

\begin{definition}
Let $S$ be an $(X,Y)$-separator of graph $G$, and let $R$ be the vertices reachable from $X$ in $G \setminus S$. Then $S$ is an \textit{important} $(X,Y)$-separator if $S$ is minimal and there is no $(X,Y)$-separator $S'$ such that $|S'| \leq |S|$ and $R' \subset R$, where $R'$ is the set of vertices reachable from $X$ in $G \setminus S'$.
\end{definition}

This definition corresponds to a ``closest'' separator, while the original definition of~\cite{M11} correspond to a ``farthest'' separator.  Important separators have been studied extensively due to their usefulness in fixed-parameter tractable algorithms, and so much is known about them.  For our purposes, we will only need the following lemma, which follows directly from Theorem 2 of~\cite{M11}.

\begin{lemma} 
\label{lem:imp_sep}
Let $X,Y \subseteq V$ be two sets of vertices in graph $G = (V,E)$, and let $d \geq 0$. An important $(X,Y)$-separator of size $d$ can be found in time $4^d \cdot n^{O(1)}$ (if one exists), where $n = |V|$.
\end{lemma}

By Lemma \ref{lem:imp_sep}, we can find an important $(X,Y)$-separator of size 2 in polynomial time. We now describe how to use this to construct what we call the \textit{$s-t$ 2-chain} of $G$. First, if there are no important $(s,t)$-separators of size 2 in $G$, then every $(s,t)$-separator has size at least 3. Hence we can just use the 2-approximation for Survivable Network Design~\cite{Jain01} with demand $(s,t,3)$ to solve the problem (or can exactly solve it by finding the cheapest three pairwise disjoint $s-t$ paths in polynomial time using a min-cost flow algorithm). 

If such an important separator exists, then we first find an important $(s,t)$-separator $S_0$ of size 2 in $G$, and let $R_0$ be the set of vertices reachable from $s$ in $G \setminus S_0$.  We let $V_{(0,r)}$ be the nodes in $R_0$ incident on $S_0$, and let $V_{(1,\ell)}$ be the nodes in $V \setminus R_0$ incident on $S_0$.  We then proceed inductively.  Given $V_{(i,\ell)}$, if there is no important $(V_{(i,\ell)}, t)$ separator of size $2$ in $G \setminus (\cup_{j=0}^{i-1} R_j)$ then the chain  is finished.  Otherwise, let $S_i$ be such a separator, let $R_i$ be the nodes reachable from $V_{(i,\ell)}$ in $(G \setminus (\cup_{j=0}^{i-1} R_j)) \setminus S_i$, let $V_{(i, r)}$ be the nodes in $R_i$ incident on $S_i$, and let $V_{(i+1, \ell)}$ be the nodes in $V \setminus (\cup_{j=0}^i R_j)$ incident on $S_i$.

After this process completes we have our $s-t$ 2-chain, consisting of components $R_0, \dots,R_p$ along with important separators $S_0, \dots, S_{p-1}$ between the components.  See Figure \ref{fig:chain}. 

\begin{figure}
    \centering
    \includegraphics[scale=1]{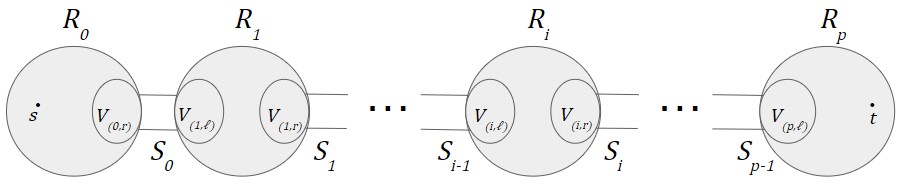}
    \caption{The $s-t$ 2-chain of $G$.}
    \label{fig:chain}
\end{figure}

We can now use this chain construction to give a structure lemma which characterizes feasible solutions. Informally, the lemma states that a subgraph $H$ of $G$ is a feasible solution if and only if in the $s-t$ 2-chain of $G$, all edges between components are in $H$, and in every component $R_i$ certain connectivity requirements between $V_{(i,\ell)}$ and $V_{(i,r)}$ are met. 

Let $G = (V,E)$ be a graph, and let $H$ be a subgraph of $G$. Going forward, we will say that in $H$, a vertex set $A \subset V$ has a path to (or is reachable from) another vertex set $B \subset V$ if there is a path from a vertex $a \in A$ to a vertex $b \in B$ in $H$. Additionally, let $X$ and $Y$ be vertex sets. We also say that $H$ satisfies the RSND demand $(X,Y,k)$ on input graph $G$ if the following is true: for every $F \subseteq E$ with $|F| < k$, if there is a path from at least one vertex in $X$ to at least one vertex in $Y$ in $G \setminus F$ then there is a path from at least one vertex in $X$ to at least one vertex in $Y$ in $H \setminus F$. The demand $(X,Y,k)$ on input $G$ is equivalent to contracting all nodes in $X$ to create super node $v_X$, contracting all nodes in $Y$ to create super node $v_Y$, and including demand $(v_X,v_Y,k)$. We will also let $G[R_i]$ and $H[R_i]$ be the subgraphs of $G$ and $H$, respectively, induced by the component $R_i$.

\begin{lemma}[Structure Lemma]
\label{lem:partition_characterize}
Let $G$ be the input graph, and let $H$ be a subgraph of $G$. Additionally, let $R_0, \dots, R_p$ denote the components in the $s-t$ 2-chain of $G$, and let $S_0, \dots, S_{p-1}$ denote the edge sets between components in the chain, as defined previously. Let $G_i = G[R_i]$, and $H_i = H[R_i]$. Then $H$ is a feasible solution to the $k=3$ Single Demand RSND problem if and only if all edges in $S_0, \dots, S_{p-1}$ are included in $H$, and $H_i$ has the following properties for every $i$:
\begin{enumerate}
    \item There are at least 3 edge-disjoint paths from $V_{(i,\ell)}$ to $V_{(i,r)}$.
    
    \item $H_i$ is a feasible solution to RSND on input graph $G_i$ with demands \[\left\{(V_{(i,\ell)},v_r,2) : v_r \in V_{(i,r)}\right\} \cup \left\{(V_{(i,r)},v_\ell,2) : v_\ell \in V_{(i,\ell)}\right\}.\]
    
    \item $H_i$ is a feasible solution to RSND on input graph $G_i$ with demands $\left\{(u,v,1) : (u,v) \in V_{(i,\ell)} \times V_{(i,r)}\right\}$.
\end{enumerate}
\end{lemma}

\iflong \else
The proof of this structure lemma is a highly technical case analysis, which due to space constraints can be found in Appendix~\ref{app:structure}.  At a very high level, though, our proof is as follows.  For the ``only if'' direction, we first assume that we are given some feasible solution $H$.  Then for each of the properties in Lemma~\ref{lem:partition_characterize}, we assume it is false and derive a contradiction by finding a fault set $F \subseteq E$  with $|F| \leq 2$ where there is a path from $s$ to $t$ in $G \setminus F$, but not in $H \setminus F$.  The exact construction of such an $F$ depends on which of the properties of Lemma~\ref{lem:partition_characterize} we are analyzing.  

For the more complicated ``if'' direction, we assume that $H$ satisfies the conditions of Lemma~\ref{lem:partition_characterize} and consider a fault set $F \subseteq E$ with $|F| \leq 2$ where $s$ and $t$ are connected in $G \setminus F$.  We want to show that $s$ and $t$ are connected in $H \setminus F$.  We analyze two subchains of the $s-t$ 2-chain of $G$: the minimal prefix of the chain which contains at least $1$ fault, and the minimal prefix of the chain which contains both faults. 
We first show that the set of vertices reachable from $s$ at the end of the first subchain is the same in $G \setminus F$ and in $H \setminus F$. We then use this to show that there is at least one reachable vertex at the end of the second subchain in $H \setminus F$, even though (unlike the first subchain) the set of reachable vertices at the end of the second subchain may be smaller in $H \setminus F$ than in $G \setminus F$. From there we show that there is a path to $t$ in $H \setminus F$ from this one reachable vertex.  There are a large number of cases depending on the structure of $F$ (whether it intersects some of the separators in the chain, whether both faults are in the same component, etc.), and we have to use different properties of Lemma~\ref{lem:partition_characterize} in different cases, making this proof technically involved.
\fi

\subsection{Algorithm and Analysis} \label{sec:3RSND-alg}
We can now use Lemma~\ref{lem:partition_characterize} to give a $7-\frac14 = \frac{27}{4}$-approximation algorithm for the $k=3$ setting of Single Demand RSND on 2-connected graphs which, by Theorem~\ref{thm:2connected}, gives a $\frac{27}{4}$-approximation algorithm for the $k=3$ Single Demand RSND problem on general graphs.  \iflong \else All missing proofs can be found in Appendix~\ref{app:3RSND-alg}.\fi

Our algorithm uses a variety of subroutines, including an algorithm for min-cost flow, the 2-RSND approximation algorithm of Theorem~\ref{thm:2RSND}, and a Steiner Forest approximation algorithm.  For reference, we state the latter of these. 

\begin{lemma}[\cite{AKR95}]
\label{lem:steiner_forest}
There is a $\left(2-\frac{1}{k}\right)$-approximation algorithm for the Steiner Forest problem, where $k$ is the number of terminal pairs in the input.
\end{lemma}

We can now give our algorithm.  Given a graph $G = (V, E)$ with edge weights $w : E \rightarrow \mathbb{R}_{\geq 0}$ and demand $\{(s,t,3)\}$, we first create the $s-t$ 2-chain of $G$ in polynomial time, as described in Section \ref{sec:decomposition}. After building the chain, within each component we run a set of algorithms to satisfy the demands characterized by Lemma~\ref{lem:partition_characterize}: a combination of min-cost flow, 2-RSND, and Steiner Forest algorithms. We include the outputs of these algorithms in our solution $H$, together with all edges in the separators $S = S_1 \cup S_2 \cup \dots \cup S_{p-1}$.

We first create an instance of min-cost flow on $G[R_i]$ (in polynomial time). Contract the vertices in $V_{(i,\ell)}$ and contract the vertices in $V_{(i,r)}$ to create super nodes $v_\ell$ and $v_r$, respectively. Let $v_\ell$ be the source node and $v_r$ be the sink node. For each edge $e \in E(R_i)$ set the capacity of $e$ to 1 and set the cost of $e$ to $w(e)$. Require a minimum flow of $3$, and run a polynomial-time min-cost flow algorithm on this instance~\cite{CLRS}. Since all capacities are integers the algorithm will return an integral flow, so we add to $H$ all edges with non-zero flow. 

We then create our first instance of 2-RSND on $G[R_i]$. Contract the vertices in $V_{(i,\ell)}$ to create super node $v_\ell$, and set demands $\{(v_\ell, u, 2) : u \in V_{(i,r)} \}$. For our second instance of 2-RSND on $R_i$, contract $V_{(i, r)}$ to create super node $v_r$, and set demands $\{(u, v_r, 2) : u \in V_{(i,\ell)} \}$. We run the 2-RSND algorithm (Theorem~\ref{thm:2RSND}) on each of these instances and include all selected edges in $H$. 

Finally, we create an instance of the Steiner Forest problem on $G[R_i]$. For each vertex pair $(v_{\ell},v_{r}) \in V_{(i,\ell)} \times V_{(i,r)}$, we check in polynomial time if $v_{\ell}$ and $v_{r}$ are connected in $G[R_i]$.  If they are connected, then we include $(v_{\ell}, v_{r})$ as a terminal pair in the Steiner Forest instance. Additionally, for $e \in E(R_i)$, we set the cost of $e$ to $w(e)$. We run the Steiner Forest approximation algorithm (Lemma~\ref{lem:steiner_forest}) on this instance, and add all selected edges to $H$.

The following lemma is essentially directly from Lemma~\ref{lem:partition_characterize} (the structure lemma) and the description of our algorithm.
\begin{lemma} \label{lem:3RSND-feasible}
$H$ is a feasible solution.
\end{lemma}

Let $H^*$ denote the optimal solution, and for any set of edges $A \subseteq E$, let $w(A) = \sum_{e \in A} w(e)$.  The next lemma follows from combining the approximation ratios of each of the subroutines used in our algorithm.

\begin{lemma} \label{lem:3RSND-cost}
$w(H) \leq \frac{27}{4} \cdot w(H^*)$
\end{lemma}

Theorem~\ref{thm:3RSND} is directly implied by Lemmas~\ref{lem:3RSND-feasible} and \ref{lem:3RSND-feasible} together with the obvious observation that our algorithm runs in polynomial time.

\bibliography{refs}

\appendix

\section{Counterexamples from Section~\ref{sec:kEFTS}} \label{app:counterexamples}
We show some counterexample to obvious approaches to $k$-EFTS; in particular, we show that our cut requirement function $f_F$ is not weakly supermodular, and the most obvious cut requirement function $f(S) = \min(k, |\delta_G(S)|)$ is also not weakly supermodular.  

Recall that $\delta_G(S)$ denotes the edges in $G$ with exactly one endpoint in $S$.  We extend this notation for disjoint sets $A, B$ by letting $\delta_G(A,B)$ denote the edges with one endpoint in $A$ and one endpoint in $B$.

\begin{theorem}
The function $f_F$ is not weakly supermodular.
\end{theorem}
\begin{proof}
Consider the following example.  Set $k = 100$.  We create a graph $G = (V, E)$ which has two sets $A, B \subseteq V$ with the following properties.
\begin{align*}
    |\delta_G(A \setminus B, V \setminus (A \cup B))|&= 49 & |\delta_G(B \setminus A, V \setminus (A \cup B))| &= 105 \\
    |\delta_G(A \cap B, V \setminus (A \cup B))| &=3 & |\delta_G(A \setminus B, B \setminus A)| &= 0 \\
    |\delta_G(A \setminus B, A \cap B)| &= 2 & |\delta_G(B \setminus A, A \cap B)| &= 49
\end{align*}
Anything not specified is extremely dense and well-connected, so an edge is in $F$ if and only if it is part of a small cut made up of the above sets.  It is not hard to see that the small cuts are precisely $A \setminus B$ (since $|\delta_G(A \setminus B)| = 49 + 0 + 2 = 51 < 100$) and $A \cap B$ (since $|\delta_G(A \cap B) = 3 + 2 + 49 = 54 < 100$).  All other cuts are large.  Hence $F$ consists of all edges involving $A$ or $B$ other than $\delta_G(B \setminus A, V \setminus (A \cup B))$, or more specifically,
\begin{align*}
F &= \delta_G(A \setminus B, V \setminus (A \cup B)) \cup \delta_G(A \cap B, V \setminus (A \cup B)) \cup \delta_G(A \setminus B, A \cap B) \cup \delta_G(B \setminus A, A \cap B).
\end{align*}
We can now calculate $f_F$ on the subsets we care about:
\begin{align*}
    f_F(A) &= 100 - 49 - 3 - 49 = -1 \\
    f_F(B) &= 100 - 3 - 2 = 95 \\
    f_F(A \setminus B) &= 0 \tag{$A \setminus B$ is small} \\
    f_F(B \setminus A) &= 100 - 49 = 51 \\
    f_F(A \cap B) &= 0 \tag{$A \cap B$ is small} \\
    f_F(A \cup B) &= 100 - 49 - 3 = 48
\end{align*}

Thus
\begin{align*}
    f_F(A) + f_F(B) &= 94 & f_F(A \setminus B) + f_F(B \setminus A) &= 51 & f_F(A \cup B) + f_F(A \cap B) &= 48
\end{align*}
Hence $f_F$ is not weakly supermodular.
\end{proof}

Note that the above example is not a contradiction of $f$ being \emph{locally} weakly supermodular since $A$ is an empty cut.

\begin{theorem}
The function $f = \min(k, |\delta_G(S)|)$ is not weakly supermodular.
\end{theorem}
\begin{proof}
Consider the following example.  Set $k = 100$.   We create a graph $G = (V, E)$ which has two sets $A, B \subseteq V$ with the following properties.  All of $A \setminus B$ and $B \setminus A$ and $A \cap B$ and $V \setminus (A \cup B)$ are extremely large and dense (e.g., large cliques).  There are no edges between $A \setminus B$, $B \setminus A$, or $A \cap B$.  The other cut sizes are:
\begin{align*}
    |\delta_G(A \cap B, V \setminus (A \cup B))| &= 55 \\
    |\delta_G(A \setminus B, V \setminus (A \cup B))| &= 95 \\
    |\delta_G(B \setminus A), V \setminus (A \cup B))| &= 95
\end{align*}
Then it is easy to see that
\begin{align*}
    f(A) &= 100 & f(b) &= 100 \\
    f(A \setminus B) &= 95 & f(B \setminus A) &= 95 \\
    f(A \cup B) &= 100 & f(A \cap B) &= 55
\end{align*}
Hence $f$ is not weakly supermodular.
\end{proof}

\iflong \else
\section{Proofs from Section~\ref{sec:kEFTS}} \label{app:kEFTS}
\begin{proof}[Proof of Lemma~\ref{lem:relaxation}]
Clearly $0 \leq x_e \leq 1$ for all $e \in E \setminus F'$.  Consider some $S \subseteq V$.  Since $H$ is a valid $k$-EFTS, the number of edges in $H \cap \delta_G(S)$ is at least $\min(k, |\delta_G(S)|)$ (or else the edges in $H \cap \delta_G(S)$ would be a fault set of size less than $k$ such that the connected components of $H$ post-faults are different from the connected components of $G$ post-faults).  Hence 
\begin{align*}
    \sum_{e \in \delta_G(S) \setminus F'} x_e &=|(H \cap \delta_G(S)) \setminus F'| = |H \cap \delta_G(S)| - |H \cap \delta_G(S) \cap F'| \geq |H \cap \delta_G(S)| - |\delta_G(S) \cap F'| \\
    &\geq \min(k, |\delta_G(S)|) - |\delta_G(S) \cap F'| = f_{F'}(S),
\end{align*}
as required.
\end{proof}

\begin{proof}[Proof of Lemma~\ref{lem:relaxation2}]
Suppose for contradiction that $H$ is not a valid $k$-EFTS.  Then there are two nodes $u,v \in V$ and a minimal set $A \subseteq E$ with $|A| < k$ so that $u,v$ are not connected in $H \setminus A$ but are connected in $G \setminus A$.  Let $S$ be the nodes reachable from $u$ in $G \setminus A$, and so by minimality of $A$ we know that $A = H \cap \delta_G(S)$.

Note that $|\delta_G(S)| > k$, or else all edges of $\delta_G(S)$ would be in $F$, implying that $E \cap \delta_G(S) = H \cap \delta_G(S) = A$ and so $u$ and $v$ would not be connected in $G \setminus A$.  Thus
\begin{align*}
    \sum_{e \in \delta_G(S) \setminus F'} x_e &= |H \cap \delta_G(S)| - |F' \cap \delta_G(S)| = |A| - |\delta_G(S) \cap F'| \\
    &< \min(k, |\delta_G(S)|) - |\delta_G(S) \cap F'| = f_{F'}(S),
\end{align*}
which contradicts $x$ being a feasible solution to LP($F'$).
\end{proof}

\begin{proof}[Proof of Lemma~\ref{lem:solve}]
We give a separation oracle, which when combined with the Ellipsoid algorithm implies the lemma~\cite{GLS88}.  Consider some vector $x$ indexed by edges of $E \setminus F'$.  Suppose that $x$ is not a feasible LP solution, so we need to find a violated constraint.  Obviously if there is some $x_e \not\in [0,1]$ then we can find this in linear time.  So without loss of generality, we may assume that there is some $S \subseteq V$ such that $\sum_{e \in \delta_G(S) \setminus F'} x_e < f_{F'}(S)$.  This implies that $f_{F'}(S) > 0$ and that there is some edge $e^* \in \delta_G(S) \setminus F'$ with $x_{e^* }< 1$ (since otherwise the LP would not be satisfiable, contradicting Lemma~\ref{lem:relaxation} and the fact that $G$ itself is a valid $k$-EFTS). Let $e^* = \{u,v\}$.  Since $e^* \not\in F'$, and $F \subseteq F'$, we know that $e^*$ cannot be part of any cuts in $G$ of size at most $k$, and thus the minimum $u-v$ cut in $G$ has more than $k$ edges.  

On the other hand, if we extend $x$ to $F'$ by setting $x_e = 1$ for all $e \in F'$, then since $S$ is a violated constraint we have that
\begin{align*}
    \sum_{e \in \delta_G(S)} x_e &= \sum_{e \in \delta_G(S) \setminus F'} x_e + |F' \cap \delta_G(S)| < f_{F'}(S) + |F \cap \delta_G(S)| \\
    &= \min(k, |\delta_G(S)|) - |\delta_G(S) \cap F'| + |\delta_G(S) \cap F'| \\
    &= k.
\end{align*}
Thus if we interpret $x$ as edge weights (with $x_e = 1$ for all $e \in F$), if we compute the minimum $s-t$ cut we will find a cut $S'$ with more than $k$ edges (since all $u-v$ cuts have more than $k$ edges) with total edge weight strictly less than $k$.  Let $S'$ be this cut.  Thus $\sum_{e \in \delta_G(S') \setminus F'} x_e < k - |\delta_G(S') \setminus F'| = f_{F'}(S')$, so $S'$ is also a violated constraint.

Hence for our separation oracle we simply compute a minimum $s-t$ cut using $x$ as edge weights for all $s,t \in V$, and if any cut we finds corresponds to a violated constraint then we return it.  By the above discussion, if there is some violated constraint then this procedure will find some violated constraint.  Thus this is a valid separation oracle.
\end{proof}

\begin{lemma}
\label{lem:big_sets}
Let $F' \supseteq F$.  If $A$ and $B$ are nonempty cuts for $f_{F'}$, then either $A \setminus B$ and $B \setminus A$ are nonempty cuts, or $A \cap B$ and $A \cup B$ are nonempty cuts.
\end{lemma}
\begin{proof}
Let 
\begin{align*}
S_1 &= \delta_{G}(A \setminus B, V \setminus (A \cup B)), &S_2 &= \delta_{G}(A \setminus B, B \setminus A), &S_3 &= \delta_{G}(A \setminus B, A \cap B), \\
S_4 &= \delta_{G}(B \setminus A,V \setminus (A \cup B)), &S_5 &= \delta_{G}(B \setminus A, A \cap B), &S_6 &= \delta_{G}(A \cap B,V \setminus (A \cup B)).
\end{align*}
Suppose that $A \setminus B$ and $A\cap B$ are both empty cuts. Each edge in $\delta_{G}(A)$ is in $S_1$, $S_2$, $S_5$, or $S_6$. Additionally, $S_1$ and $S_2$ are subsets of $\delta_{G}(A \setminus B)$, while $S_5$ and $S_6$ are subsets of $\delta_{G}(A \cap B)$. This means that every edge in $\delta_{G}(A)$ is in an empty cut, and so all edges in $\delta_{G}(A)$ are in $F'$. Thus $A$ is an empty cut, contradicting the assumption of the lemma.  Thus at least one of $A \setminus B$ and $A \cap B$ is nonempty.  If we instead assume that $B \setminus A$ and $A\cap B$ are empty cuts, then we can use a similar argument to prove that $B$ is an empty cut. This proves that at least one of $B \setminus A$ and $A\cap B$ are nonempty.  Hence if $A \cap B$ is empty, then both $A \setminus B$ and $B \setminus A$ are nonempty, proving the lemma.  

Now suppose that $A \setminus B$ and $A\cup B$ are both empty cuts. Each edge in $\delta_{G}(B)$ is in $S_2$, $S_3$, $S_4$, or $S_6$. Additionally, $S_2$ and $S_3$ are subsets of $\delta_{G}(A \setminus B)$, while $S_4$ and $S_6$ are subsets of $\delta_{G}(A \cup B)$. This means that every edge in $\delta_{G}(B)$ is in an empty cut, and so all edges in $\delta_{G}(B)$ are in $F'$.  Thus $B$ is an empty cut, contradicting the assumption of the lemma.  Thus at least one of $A \setminus B$ and $A \cup B$ is nonempty.  If we instead assume that $B \setminus A$ and $A \cup B$ are empty cuts, then we can use a similar argument to prove that $A$ is empty, and hence at least one of $B \setminus A$ and $A \cup B$ is nonempty.  Hence if $A \cup B$ is empty, then both $A \setminus B$ and $B \setminus A$ are nonempty, proving the lemma.

Thus either both $A \setminus B$ and $B \setminus A$ are nonempty, or both $A \cap B$ and $A \cup B$ are nonempty, proving the lemma.  
\end{proof}

\begin{proof}[Proof of Theorem~\ref{thm:LWS}]
Let $F' \supseteq F$, and suppose $A$ and $B$ are nonempty cuts. Let 
\begin{align*}
S_1 &= \delta_{G}(A \setminus B, V \setminus (A \cup B)), &S_2 &= \delta_{G}(A \setminus B, B \setminus A), &S_3 &= \delta_{G}(A \setminus B, A \cap B), \\
S_4 &= \delta_{G}(B \setminus A,V \setminus (A \cup B)), &S_5 &= \delta_{G}(B \setminus A, A \cap B), &S_6 &= \delta_{G}(A \cap B,V \setminus (A \cup B)).
\end{align*}
We also let $s_i = |S_i \cap F'|$ for $i \in [6]$.

$A$ and $B$ are nonempty cuts, so $A$ and $B$ must be large cuts and $\min(k,|\delta_{G}(A)|) = \min(k,|\delta_{G}(B)|) = k$. Each edge in $\delta_{G}(A)$ is in exactly one of $S_1$, $S_2$, $S_5$, and $S_6$, and each edge in $\delta_{G}(B)$ is in exactly one of $S_2$, $S_3$, $S_4$, and $S_6$, so we have that $|\delta_{G}(A) \cap F'| = s_1 + s_2 + s_5 + s_6$ and $|\delta_{G}(B) \cap F'| = s_2 + s_3 + s_4 + s_6$. We therefore have the following:
\begin{align}
    &f_{F'}(A) = \min(k,|\delta_{G}(A)|) - |\delta_{G}(A) \cap F| =  k - s_1 - s_2 - s_5 - s_6 \notag \\
    &f_{F'}(B) = \min(k,|\delta_{G}(B)|) - |\delta_{G}(B) \cap F| = k - s_2 - s_3 - s_4 - s_6 \notag \\
    \implies &f_{F'}(A) + f_{F'}(B) = 2k - s_1 - 2s_2 - s_3 - s_4 - s_5 - 2s_6. \label{eq:AB}
\end{align}

$A$ and $B$ are nonempty so by Lemma~\ref{lem:big_sets}, either $A \setminus B$ and $B \setminus A$ are nonempty cuts, or $A \cap B$ and $A \cup B$ are nonempty cuts. Suppose first that $A \setminus B$ and $B \setminus A$ are nonempty cuts, which implies that $\min(k,|\delta_{G}(A \setminus B)|) = \min(k,|\delta_{G}(B \setminus A)|) = k$. Each edge in $\delta_{G}(A \setminus B)$ is in exactly one of $S_1$, $S_2$, and $S_3$, and each edge in $\delta_{G}(B \setminus A)$ is in exactly one of $S_2$, $S_4$, and $S_5$, so we have that $|\delta_{G}(A \setminus B) \cap F'| = s_1 + s_2 + s_3$ and $|\delta_{G}(B \setminus A) \cap F'| = s_2 + s_4 + s_5$. Putting this all together, we get the following for $f_{F'}(A \setminus B)$ and $f_{F'}(B \setminus A)$:
\begin{align*}
    &f_{F'}(A \setminus B) = \min(k,|\delta_{G}(A \setminus B)|) - |\delta_{G}(A \setminus B) \cap F'| = k - s_1 - s_2 - s_3 \\
    &f_{F'}(B \setminus A) = \min(k,|\delta_{G}(B \setminus A)|) - |\delta_{G}(B \setminus A) \cap F'| = k - s_2 - s_4 - s_5 \\
    \implies &f_{F'}(A \setminus B) + f(B \setminus A) = 2k - s_1 - 2s_2 - s_3 - s_4 - s_5.
\end{align*}
This and~\eqref{eq:AB} imply that $f_{F'}(A) + f_{F'}(B) \leq f_{F'}(A\setminus B) + f_{F'}(B\setminus A)$ if $A \setminus B$ and $B \setminus A$ are nonempty cuts. 

Now suppose that $A \cap B$ and $A \cup B$ are nonempty cuts, and so $\min(k,|\delta_{G}(A \setminus B)|) = \min(k,|\delta_{G}(B \setminus A)|) = k$. Each edge in $\delta_{G}(A \cap B)$ is in exactly one of $S_3$, $S_5$, and $S_6$, and each edge in $\delta_{G}(A \cup B)$ is in exactly one of $S_1$, $S_4$, and $S_6$, so we have that $|\delta_{G}(A \cap B) \cap F'| = s_3 + s_5 + s_6$ and $|\delta_{G}(A \cup B) \cap F'| = s_1 + s_4 + s_6$. Putting this all together, we get the following for $f_{F'}(A \cap B)$ and $f_{F'}(A \cup B)$:

\begin{align*}
    &f_{F'}(A \cap B) = \min(k,|\delta_{G}(A \cap B)|) - |\delta_{G}(A \cap B) \cap F'| = k - s_3 - s_5 - s_6 \\
    &f_{F'}(A \cup B) = \min(k,|\delta_{G}(A \cup B)|) - |\delta_{G}(A \cup B) \cap F'| = k - s_1 - s_4 - s_6 \\
    \implies &f_{F'}(A \cap B) + f_{F'}(A \cup B) = 2k - s_1 - s_3 - s_4 - s_5 - 2s_6.
\end{align*}
This and~\eqref{eq:AB} imply that $f_{F'}(A) + f_{F'}(B) \leq f_{F'}(A\cap B) + f_{F'}(B\cup A)$ if $A \cap B$ and $A \cup B$ are nonempty cuts.
\end{proof} 

\begin{proof}[Proof of Lemma~\ref{lem:tight-dimension}]
Let $\mathcal L$ be a maximal laminar family of tight sets.  Lemma~\ref{lem:laminar-span} implies that $\Span(\mathcal L) = \Span(\mathcal T)$, so it suffices to upper bound the number of sets in $\mathcal L$.  And since we care about the span, if there are two sets $S, S'$ with $\mathcal A_G(S) = \mathcal A_G(S')$ then we can remove one of them from $\mathcal L$ arbitrarily, so no two sets in $\mathcal L$ have identical rows in the constraint matrix.  

Any set that consists of exclusively low degree nodes cannot be tight, since the set has no corresponding row in the constraint matrix. Thus, all sets in $\mathcal L$ must contain at least one high degree node, and hence all minimal sets in $\mathcal L$ have at least one high degree node. 

Let $S \in \mathcal L$, and let $S' \supset S$ so that every node in $S' \setminus S$ is a low-degree node.  Then every edge edge in $(\delta_G(S) \setminus \delta_G(S')) \cup (\delta_G(S') \setminus \delta_G(S))$ must be incident on at least one low-degree node and hence is in $F$.  Thus $\mathcal A_G(S) = \mathcal A_G(S')$, and hence $S'$ is not in $\mathcal L$.  Therefore, any superset $S'$ in the laminar family of some other set $S$ in the laminar family must have at least one more high degree node than $S$.  

Since any minimal set in $\mathcal L$ has at least one high degree node, and every set in $\mathcal L$ contains at least one more high degree node than any set in $\mathcal L$ that it contains, if we restrict each set in $\mathcal L$ to the high-degree nodes then we have a laminar family on the high-degree nodes.  Thus $|\mathcal L| \leq 2n_h - 1$.
\end{proof}

\section{Reduction to RSND on 2-Connected Graphs} \label{app:reduction}
In this section, we give a reduction from the Relative Survivable Network Design (RSND) problem on general graphs to the RSND problem on $2$-connected graphs. We then use this reduction to give a $2$-approximation algorithm for the special case of RSND in which all demands are at most $2$.

\subsection{Definitions}
Let $G'=(V,E')$ be the subgraph of $G$ obtained by removing all edges in cuts of size 1 from $G$.  We now construct the \emph{component graph} $G_C$ as follows.

\begin{definition}
Let $G_C = (V_C, E_C)$ be a component graph, where each connected component $C \in G'$ is represented by a vertex $v_C \in V_C$. Let $C_i$ and $C_j$ be connected components in $G'$. The edge $(v_{C_i}, v_{C_j})$ is in $E_C$ if and only if there exists vertices $i \in C_i$ and $j \in C_j$, such that $(i,j) \in E$.
\end{definition}

It is easy to see that $G_C$ is a tree and that every connected component of $G'$ is $2$-edge connected. 

\begin{definition}
A vertex $t \in V$ is a \emph{terminal vertex} if $t$ is adjacent to at least one edge in $E \setminus E'$. For each terminal vertex $t$, let $P_t$ be the set of vertex pairs, $(u,v)$, such that $u$ and $v$ are in different connected components in $G'$, and such that every $u-v$ path uses an edge in $E \setminus E'$ that has $t$ as an endpoint.
\end{definition}

\subsection{Reduction}
We are now able to give a reduction to RSND on 2-connected graphs. Going forward, it will be easier to refer to RSND demands using a demand function. We say that an RSND instance on graph $G$ with demands $\{(s_i, t_i, k_i)\}_{i \in [\ell]}$ has a corresponding demand function, $r : V \times V \rightarrow \mathbb{Z}$, such that $r(s_i,t_i) = k_i$ for all $i$ and $r(u,v) = 0$ for all other pairs.

\paragraph{Reduction:}  We reduce from an RSND instance on input graph $G = (V,E)$ with demand function $r(u,v)$ on vertex pairs $u,v \in V$ and edge weights $w : E \rightarrow \mathbb{R}_{\geq 0}$ to a new instance of RSND. The new instance is on graph $G_R$, has edge weight function $w_R$, and a demand function $r_R(u,v)$.  The input graph and edge weight function are unchanged: We set $G_R = G$ and $w_R = w$. Now we define $r_R(u,v)$. For each connected connected component $C \in G'$, the reduction is as follows:
\begin{enumerate}
    \item For each vertex pair $u,v \in C$ such that $u$ and $v$ are not terminal vertices, set $\displaystyle r_R(u,v) = r(u,v)$
    \item For each vertex pair $u,t \in C$ such that $u$ is not a terminal vertex and $t$ is a terminal vertex, set  $\displaystyle r_R(u,t) =  \max\left\{ r(u,t), \max_{v : (u,v) \in P_t }r(u,v)\right\}$
    \item For each vertex pair $t_1, t_2 \in C$ such that $t_1$ and $t_2$ are terminal vertices, set \\ $\displaystyle r_R(t_1,t_2) = \max \left\{r(t_1,t_2), \max_{(v,w) \in P_{t_1} \cap P_{t_2} } r(v,w) \right\}$.
\end{enumerate}
For all vertex pairs $u,v$ such that $u$ and $v$ are in different connected components in $G'$, if $r(u,v) > 0$ then we set $r_R(u,v) = 1$.

\begin{lemma}
Any feasible solution to the RSND problem on input graph $G$ with edge weight function $w(e)$ and demand function $r(u,v)$ is also a feasible solution to the RSND problem on input graph $G_R$ with edge weight function $w_R(e)$ and demand function $r_R(u,v)$.
\end{lemma}
\begin{proof}
We show that given a feasible subgraph $H_A$ to the original instance, $H_A$ is also a feasible solution to the reduction instance (and thus has the same cost). In particular, we will show that for each vertex pair $u,v \in V$, for any edge fault set $F$ with $|F| < r_R(u,v)$, $u$ and $v$ are connected in $G \setminus F$ if and only if they are connected in $H_A \setminus F$. When this property holds for a fixed vertex pair $u,v$ in $H_A$, we say that $u$ and $v$ are relative fault tolerant with respect to $G$ under the reduction instance (that is, with demand function $r_R(u,v)$). We will show that $u$ and $v$ are relative fault tolerant with respect to $G$ under the reduction instance. We have the following cases:
\begin{enumerate}
    \item Vertices $u$ and $v$ are in the same connected component $C$ in $G'$, and neither $u$ nor $v$ is a terminal vertex. Then, $r_R(u,v) = r(u,v)$. Both RSND instances have the same demand for the vertex pair. The subgraph $H_A$ is a feasible solution to the original instance, so $u$ and $v$ are relative fault tolerant with respect to $G$ under the original instance. Therefore, $u$ and $v$ in must still be relative fault tolerant with respect to $G_R$ under the reduction instance in this case.
    
    \item Vertices $u$ and $v$ are in the same connected component $C$ in $G'$, and exactly one of $u$ and $v$ is a terminal vertex. Suppose without loss of generality that $v$ is the terminal vertex. First, suppose that $r_R(u,v) = r(u,v)$. Both instances have the same demand for this vertex pair, so the argument is identical to that given in Case 1. Now suppose that $r_R(u,v) = \max_{x : (u,x) \in P_v }r(u,x)$. Let $x = \argmax_{x : (u,x) \in P_v }r(u,x)$.
    \begin{itemize}
        \item Suppose $u$ and $x$ are connected in $G \setminus F$, where $|F| < r(u,x) = r_R(u,v)$. Vertices $u$ and $x$ are in different connected components in $G'$, and every $u-x$ path must use an edge in $E \setminus E'$ that has $v$ as an endpoint. Therefore, $u$ and $v$ must also be connected in $G \setminus F$. Since $H_A$ is a feasible solution to the original instance, we have that if $u$ and $x$ are connected in $G \setminus F$ for some fault set $F$ with $|F| < r(u,x)$, then $u$ and $x$ are also connected in $H_A \setminus F$. Combining this with the fact that a path from $u$ to $x$ implies a path from $u$ to $v$ gives us the following:  If $u$ and $x$ are connected in $G \setminus F$ for some fault set $F$ with $|F| < r(u,x) = r_R(u,v)$, then $u$ and $v$ are connected in both $G \setminus F$ and in $H_A \setminus F$. Therefore, since $u$ and $x$ are connected in $G \setminus F$ (and in $H_A \setminus F$), $u$ and $v$ must be connected in $H_A \setminus F$. This means that $u$ and $v$ are relative fault tolerant with respect to $G$ under the reduction instance in this case. 
        
        \item Now suppose $u$ and $x$ are not connected in $G \setminus F$, but $u$ and $v$ are still connected in $G \setminus F$, for some fault set $F$ with $|F| < r(u,x) = r_R(u,v)$. Since $G_C$ is a tree and $u$ and $v$ are in the same connected component, $C \in G'$, vertices $u$ and $v$ can only be separated in $H_A$ by edges in $E(C)$. Therefore we only consider the edges in $F$ that are in $E(C)$. Let $F_C = E(C) \cap F$, and note that $u$ and $v$ are connected in $G \setminus F_C$. We will show that $u$ and $v$ must also be connected in $H_A \setminus F_C$ (and therefore in $H_A \setminus F$). Vertices $u$ and $v$ are connected in $G \setminus F_C$, and $F_C \subseteq E(C)$; therefore, $u$ and $x$ are also connected in $G \setminus F_C$. Additionally, since $|F_C| < r(u,x)$, we have the following: If $u$ and $x$ are connected in $G \setminus F_C$, then $u$ and $x$ are also connected in $H_A \setminus F_C$. This implies that $u$ and $v$ are connected in $H_A \setminus F_C$, and therefore in $H_A \setminus F$, meaning that $u$ and $v$ are relative fault tolerant with respect to $G$ under the reduction instance in this case. 
    \end{itemize}
    
    \item Vertices $u$ and $v$ are in the same connected component $C$ in $G'$, and both $u$ and $v$ are terminal vertices. First, suppose that $r_R(u,v) = r(u,v)$. Both instances have the same demand for the vertex pair, and so the argument is identical to that given in Case 1. Now suppose that $r_R(u,v) = \max_{(x,y) \in P_{u} \cap P_{v} } r(x,y)$. Let $(x,y) =  \argmax_{(x,y) \in P_{u} \cap P_{v} } r(x,y)$.
    \begin{itemize}
        \item Suppose $x$ and $y$ are connected in $G \setminus F$, where $|F| < r(x,y) = r_R(u,v)$. Vertices $x$ and $y$ are in different connected components in $G'$, and every $x-y$ path must use an edge that has $u$ as an endpoint and an edge that has $v$ as an endpoint. Therefore, $u$ and $v$ must also be connected in $G \setminus F$. Additionally, $H_A$ is a feasible solution to the original instance, so we have that if $x$ and $y$ are connected in $G \setminus F$ for some fault set $F$ with $|F| < r(x,y)$, then $x$ and $y$ are also connected in $H_A \setminus F$. Combining this with the fact that a path from $x$ to $y$ implies a path from $u$ to $v$ gives us the following:  If $x$ and $y$ are connected in $G \setminus F$ for some fault set $F$ with $|F| < r(x,y) = r_R(u,v)$, then we have that $u$ and $v$ are also connected in both $G \setminus F$ and in $H_A \setminus F$. Therefore, $u$ and $v$ must be connected in $H_A \setminus F$, and so $u$ and $v$ are relative fault tolerant under the reduction instance in this case.
        
        \item Now consider the case when $x$ and $y$ are not connected in $G \setminus F$, but $u$ and $v$ are connected in $G \setminus F$, for some fault set $F$ with $|F| < r(x,y) = r_R(u,v)$. Since $G_C$ is a tree, $u$ and $v$ can only be separated in $H_A$ by edges in $E(C)$. Therefore, we only consider the edges in $F$ that are in $E(C)$. Let $F_C = E(C) \cap F$, and note that $u$ and $v$ are connected in $G \setminus F_C$. We will show that $u$ and $v$ must also be connected in $H_A \setminus F_C$ (and therefore in $H_A \setminus F$). Vertices $u$ and $v$ are connected in $G \setminus F_C$, and $F_C \subseteq E(C)$; therefore, $x$ and $y$ are also connected in $G \setminus F_C$. Additionally, because $|F_C| < r(x,y)$, we have that if $x$ and $y$ are connected in $G \setminus F_C$, then $x$ and $y$ are also connected in $H_A \setminus F_C$. This implies that $u$ and $v$ are connected in $H_A \setminus F_C$, and therefore in $H_A \setminus F$. Thus, $u$ and $v$ are relative fault tolerant with respect to $G$ under the reduction instance in this case.
    \end{itemize}
    
    \item Vertices $u$ and $v$ are in different connected components in $G'$. There is only one edge-disjoint path from $u$ to $v$ in $G$. In the reduction instance, if $r_R(u,v) = 1$ then $r(u,v) > 0$.  Since $H_A$ is feasible, if $r(u,v) > 0$, then there is a single path from $u$ to $v$ in $H_A$ if there is a path from $u$ to $v$ in $G$. Therefore, the demand $r_R(u,v) = 1$ is always satisfied in $H_A$, and so $u$ and $v$ are relative fault tolerant with respect to $G$ under the reduction instance.
\end{enumerate}
\end{proof}

We now show that a feasible solution to the reduction RSND instance is a feasible solution to the original instance.

\begin{lemma}
\label{lem:reduction}
Any feasible solution to the RSND problem on input graph $G_R$ with edge weight function $w_R(e)$ and demand function $r_R(u,v)$ is also a feasible solution to the original RSND problem on input graph $G$ with edge weight function $w(e)$ and demand function $r(u,v)$.
\end{lemma}
\begin{proof}
We show that given a feasible solution subgraph $H_B$ to the reduction instance, $H_B$ is also a feasible solution to the original RSND instance (and thus has the same cost). If vertices $u$ and $v$ are in the same connected component in $G'$, then $r_R(u,v) \geq r(u,v)$. As a result, $u$ and $v$ are relative fault tolerant with respect to $G$ under the original instance. 

Suppose instead that $u$ and $v$ are in different connected components. Let $C_u$ and $C_v$ be different connected components in $G'$, and let $u \in C_u$ and $v \in C_v$ be vertices in these components. We will show that if $u$ and $v$ are connected in $G \setminus F$, where $F$ is an edge fault set with $|F| < r(u,v)$, then $u$ and $v$ are connected in $H_B \setminus F$.

The component subgraph $G_C$ is a tree and $u$ and $v$ are in different components, so there is a size 1 cut that separates $u$ and $v$ in $G$. Therefore, if $u$ and $v$ are connected in $G \setminus F$ for some edge fault set $F$, then $F$ cannot have an edge from any of the size 1 cuts that separate $u$ and $v$. Note that all other size 1 cuts are not on any $u$-$v$ path. As a result, we only need to consider fault sets $F$ such that $|F| > 1$ and $F$ does not contain size 1 cuts. Any such $F$ must have all edges within the connected components of $G'$. We can assume without loss of generality that all edges in $F$ are in the same connected component in $G'$. We will now show that if $u$ and $v$ are connected in $G \setminus F$, where $|F| < r(u,v)$, then $F$ cannot separate $u$ or $v$ from any of its terminal vertices in $H_B \setminus F$ (and therefore $u$ and $v$ are connected in $H_B \setminus F$).

Suppose $F$, with $|F| < r(u,v)$, is one of these fault sets, and that without loss of generality that $F \subseteq E(C_u)$ (the argument is identical when $F \subseteq E(C_v)$; we will later handle the case when $F \subseteq E(C_i)$ where $i \neq u,v$). Let $t_u \in C_u$ be the terminal vertex such that $(u,v) \in P_{t_u}$. Since $r(u,v) \leq r_R(u,t_u)$, we have that if $u$ and $t_u$ ($v$ and $t_v$) are connected in $G \setminus F$, then $u$ and $t_u$ ($v$ and $t_v$) are connected in $H_B \setminus F$.

Now suppose that $u$ and $v$ are connected through some other connected component, $C_i$, such that $u,v \notin C_i$. Also, let $t_1$ and $t_2$ be terminal vertices in $C_i$ such that $(u,v) \in P_{t_1} \cap P_{t_2}$.  Suppose in addition that $F \in C_i$, and that there is a path from $u$ to $t_1$ and a path from $v$ to $t_2$ in $G \setminus F$. If $t_1 \neq t_2$, then because $r(u,v) \leq r_R(t_1,t_2)$, we have that if $t_1$ and $t_2$ are connected in $G \setminus F$, then they are connected in $H_B \setminus F$. We have shown that if $u$ and $v$ are connected in $G \setminus F$, where $|F| < r(u,v)$, then in $H_B$, $F$ cannot separate $u$ or $v$ from any of their terminal vertices.

Finally, we consider the empty fault set. That is, we want to show that if $u$ and $v$ are connected in $G$, then they are connected in $H_B$. For every vertex pair $u,v$, if $r(u,v) > 0$, then $r_R(u,v) = 1$. Subgraph $H_B$ is feasible, so if $r_R(u,v) = 1$ then there is a path from $u$ to $v$ in $H_B$ if there is a path from $u$ to $v$ in $G$. This means that if $r(u,v) > 0$, there is a path from $u$ to $v$ in $H_B$ if there is a path from $u$ to $v$ in $G$. Since there is only 1 edge-disjoint path from $u$ to $v$ in $G$, we have that $u$ and $v$ are relative fault tolerant with respect $G$ under the original instance.
\end{proof}

We have shown that any feasible solution to one instance is also a feasible solution to the other. This also implies that the optimal solution to both the original and reduction instances is the same, and has the same value. We now give a corollary that allows us to assume that any input graph of an RSND instance is 2-edge connected.

\begin{theorem}
\label{thm:2connected-app}
If there exists an $\alpha$-approximation algorithm for RSND on 2-edge connected graphs, then there is an $\alpha$-approximation algorithm for RSND on general graphs. 
\end{theorem}
\begin{proof}
Suppose we have an $\alpha$-approximation algorithm for RSND on 2-edge connected graphs. An $\alpha$-approximation algorithm for RSND on general graphs is as follows: Perform the reduction described above, and run the $\alpha$-approximation algorithm for RSND on 2-edge connected graphs on each connected component in $G'$ (recall that each component is 2-edge connected). Then, for each edge $e \in E \setminus E'$, we include $e$ in the solution subgraph $H$ if there exists a pair of connected components $C_i, C_j \in G'$ such that $e$ is on the path from $C_i$ and $C_j$ and there exists vertices $v_i \in C_i$ and $v_j \in C_j$ such that $r_R(v_i,v_j) > 0$.

The algorithm returns a subgraph that is a feasible solution to the reduction instance, so the subgraph is a feasible solution to the original RSND instance by Lemma \ref{lem:reduction}. Now we will show that the algorithm gives an $\alpha$-approximation of the RSND problem. Let $H$ be the solution subgraph returned by the algorithm, and let $H^*$ be the optimal solution to the RSND problem instance. Let $C_1, C_2, \dots, C_\ell$ be the connected components of $G'$. For a fixed connected component $C_i$, let $H_i = H[C_i]$ be the subgraph of $H$ induced by component $C_i$, and let $H_i^* = H^*[C_i]$ be the subgraph of $H^*$ induced by $C_i$. Let $c(S)$ denote the total weight of a subgraph $S$, and $c(T)$ denote the sum of the weights of each edge in the edge set $T$. 

Each connected component under the reduction instance is an instance of the RSND problem on 2-edge connected graphs, and the algorithm runs the $\alpha$-approximation for 2-edge connected RSND on each instance. Hence $c(H_i) \leq \alpha \cdot c(H_i^*)$ for all $i$.  Summing over all connected components, we get that
\begin{align*}
     \sum_{i=1}^{\ell} c(H_i) &\leq \alpha \cdot \sum_{i=1}^{\ell} c(H_i^*).
\end{align*}

Finally, let $\tilde{E}$ be the set of edges from size 1 cuts in $G$ that are included in the algorithm solution $H$. Additionally, let $\tilde{E}^*$ be the set of edges from size 1 cuts in $G$ that are included in the optimal solution. Any edge $e$ from a size 1 cut must be included in any feasible solution if $e$ connects a vertex pair with positive demand. The algorithm only selects the edges from size 1 cuts that must be included in any feasible solution, so $\tilde{E} = \tilde{E^*}$. Putting everything together, we have the following:
\begin{align*}
    c(ALG) &= \sum_{i=1}^{\ell} c(H_i) + c(\tilde{E}) \leq \alpha  \sum_{i=1}^{\ell} c(H_i^*) +  c(\tilde{E}) \leq \alpha \left( \sum_{i=1}^{\ell} c(H_i^*) + \cdot c(\tilde{E}^*) \right) = \alpha \cdot c(OPT). \qedhere
\end{align*}
\end{proof}

\subsection{Special Case: 2-RSND}
\label{app:2-RSND}
The $k$-RSND problem is a special case of the RSND problem. In $k$-RSND, the input is still a graph $G=(V,E)$ and a demand for each vertex pair; however, all demands are at most $k$. 

\begin{theorem} \label{thm:2RSND-app}
There is a $2$-approximation algorithm for $2$-RSND.
\end{theorem}
\begin{proof}
For every vertex pair $u,v$ in a two-connected graph, the number of edge disjoint paths from $u$ to $v$ is at least $2$. Therefore, an instance of 2-connected 2-RSND is an instance of SND, and hence Jain's $2$-approximation for SND~\cite{Jain01} is 2-approximation for 2-RSND on 2-connected graphs. By Theorem~\ref{thm:2connected}, this gives a 2-approximation algorithm for 2-RSND on general graphs. 
\end{proof}

Unfortunately, this algorithm does not directly extend to larger demand upper bounds. If we require the connected components in $G'$ to be $k$-edge connected, where $k$ is the demand upper bound, then the edges in $E \setminus E'$ may not form a tree on the components of $G'$. This would require some other means for selecting edges between these connected components. If we instead carry out the reduction from Theorem~\ref{thm:2connected}, then each connected component would still be an instance of general RSND, since there may be demands that are larger than 2, while each connected component is only guaranteed to be 2-edge connected.

\section{Proof of Lemma~\ref{lem:partition_characterize} (Structure Lemma)} \label{app:structure} 
Let $i \leq j$, let $H$ be a subgraph of $G$, and let $H_i = H[R_i]$. We will say that an edge $e$ is between subgraphs $H_i$ and $H_j$, or that $e$ is within the subchain that starts at $H_i$ and ends at $H_j$, if the following is true: Either edge $e$ is in $E(H_k)$ such that $i \leq k \leq j$, or $i \neq j$ and $e \in S_k$ such that $i \leq k \leq j-1$. Before we begin the proof, we will need a lemma that describes the connectivity of vertices in $V_{(i,\ell)}$ and $V_{(j,r)}$ in $G$ when there are no edge faults between components $R_i$ and $R_j$.

\begin{lemma}
\label{lem:no_faults_G}
In the $s-t$ 2-chain of $G$, consider the subchain that starts at $G_i$ and ends at $G_j$, inclusive, where $i \leq j$. Then for every $u \in V_{(i,\ell)}$, there is a path from $u$ to $V_{(j,r)}$.  Similarly, for each $u \in V_{(j,r)}$, there is a path from $u$ to $V_{(i,\ell)}$.  These paths only use edges within the subchain that starts at $G_i$ and ends at $G_j$.
\end{lemma}
\begin{proof}
First we will show that for any subgraph $G_k$ within the subchain, where $G_k = G[R_k]$, there is a path in $G_k$ from $V_{(k,\ell)}$ to each vertex in $V_{(k,r)}$. Fix $G_k$ with $i \leq k \leq j$. Suppose for the sake of contradiction that at least one vertex $v_{kr_1} \in V_{(k,r)}$ is not reachable from $V_{(k,\ell)}$ in $G_k$. If $|V_{(k,r)}|= 1$, then there is no path from $V_{(k,\ell)}$ to $V_{(k,r)}$ in $G_k$. This means there is a size 0 cut in $G_k$, and therefore a size 0 cut in $G$.  But $G$ is 2-connected, so this gives a contradiction. If $|V_{(k,r)}|= 2$, let $V_{(k,r)} = \{ v_{kr_1}, v_{kr_2} \}$. Then, the edge in $S_k$ that is incident on $v_{kr_2}$ is a $(V_{(k,\ell)},t)$-separator in $G \setminus \cup_{j=0}^{k-1} R_j$ with size 1, meaning that $S_k$ is not minimal. This contradicts the assumption that $S_k$ is an important separator. Therefore, in $G_k$, each vertex in $V_{(k,r)}$ must be reachable from $V_{(k,\ell)}$ in $G_k$. 

Now we show via a similar argument that in $G_k$, there is a path from each vertex in $V_{(k,\ell)}$ to $V_{(k,r)}$. Fix $G_k$ with $i \leq k \leq j$. Suppose for the sake of contradiction that at least one vertex $v_{k\ell_1} \in V_{(k,\ell)}$ is not reachable from $V_{(k,r)}$ in $G_k$. If $|V_{(k,\ell)}|= 1$, then there is no path from $V_{(k,\ell)}$ to $V_{(k,r)}$ in $G_k$, meaning there is a size 0 cut in $G$. This contradicts the assumption that $G$ is 2-connected. If $|V_{(k,\ell)}|= 2$, let $V_{(k,\ell)} = \{ v_{k\ell_1}, v_{k\ell_2} \}$. Then, the edge in $S_{k-1}$ that is incident on $v_{k\ell_2}$ is a $(V_{(k-1,\ell)},t)$-separator in $G \setminus \cup_{j=0}^{k-2} R_j$ with size 1, meaning that $S_{k-1}$ is not minimal. This contradicts the assumption that $S_{k-1}$ is an important separator.

We have shown that for each subgraph $G_k$ in the subchain starting at $G_i$ and ending at $G_j$, each vertex in $V_{(k,r)}$ is reachable from $V_{(k,\ell)}$, and $V_{(k,r)}$ is reachable from each vertex in $V_{(k,\ell)}$ in $G_k$. This implies that each vertex in $V_{(j,r)}$ is reachable from $V_{(i,\ell)}$, and that $V_{(j,r)}$ is reachable from each vertex in $V_{(i,\ell)}$, using only edges in the subchain from $G_i$ to $G_j$. This can been seen via a proof by induction on the number of components into the $s-t$ 2-chain. 
\end{proof}

\subsection{Only if}
We are now ready to prove that the properties in Lemma~\ref{lem:partition_characterize} are necessary. Suppose subgraph $H$ is a feasible solution, and suppose for the sake of contradiction that for some important separator $S_i$ in the chain, there is an edge $e_1 \in S_i$ that is not in $H$. Let $e_2$ be the other edge in $S_i$ (note that $e_2$ must exist since $|S_i| = 2$ for all $i$). First, suppose that $e_2$ is in $H$. We will show that $s$ and $t$ are connected in $G \setminus \{e_2\}$ but that they are not connected in $H \setminus \{e_2\}$, giving a contradiction. Separator $S_i$ is a size 2 $s-t$ cut, so $s$ and $t$ are not connected in $H \setminus \{e_2\}$. Now we just need to show that they are connected in $G \setminus \{e_2\}$. There are no edge faults between $G_0$ and $G_i$, so by Lemma \ref{lem:no_faults_G} there is a path from $V_{(0,\ell)} = s$ to each vertex in $V_{(i,r)}$ that only uses edges between $G_0$ and $G_i$. In $G \setminus \{e_2\}$, one vertex in $V_{(i,r)}$ is adjacent to a vertex in $V_{(i+1,\ell)}$; they share the edge $e_1$. Therefore there is a path from a vertex in $V_{(i,r)}$ to a vertex in $V_{(i+1,\ell)}$ that only uses the edge $e_1$. Finally, there are no edge faults between $G_{i+1}$ and $G_p$, so again by Lemma \ref{lem:no_faults_G} there is a path from each vertex in $V_{(i+1,\ell)}$ to $V_{(p,r)} = t$, using only edges between $G_{i+1}$ and $G_p$. Putting everything together, we have that $s$ and $t$ are connected in $G \setminus \{e_2\}$ but are not connected in $H \setminus \{e_2\}$, contradicting the assumption that $H$ is feasible. Now suppose that $e_2$ is not in $H$. Then, $s$ and $t$ are connected in $G$ but not in $H$. This also contradicts the assumption that $H$ is feasible, and so $H$ must have all edges in $\cup_{j=0}^{p-1} S_j$.

Now suppose that $H$ is feasible, but suppose for the sake of contradiction that Property 1 in the statement of Lemma \ref{lem:partition_characterize} is not satisfied in $H_i$. This means there are fewer than three edge-disjoint paths from $V_{(i,\ell)}$ to $V_{(i,r)}$ in $H_i$, and so by Menger's theorem there is some cut $S_i'$ of size $2$ that separates $V_{(i,\ell)}$ and $V_{(i,r)}$ in $H_i$. This directly implies that $S_i'$ also separates $V_{(i,\ell)}$ from $t$ in $H_i$. Cut $S_i'$ is therefore a $(V_{(i,\ell)}, t)$-separator of size 2 with $R_i' \subset R_i$, where $R_i'$ is the set of vertices reachable from $V_{(i,\ell)}$ in $(G \setminus  \cup_{j=0}^{i-1} R_j) \setminus S_i'$. Separator $S_i$ is therefore not an important $(V_{(i,\ell)},t)$-separator in $G \setminus \cup_{j=0}^{i-1} R_j$, contradicting our decomposition construction.

Suppose again that $H$ is feasible, but now suppose for the sake of contradiction that Property 2 in the statement of Lemma \ref{lem:partition_characterize} is not satisfied in $H_i$. Without loss of generality
\footnote{WLOG because the argument is symmetric: If we choose $v_{\ell_1} \in V_{(i,\ell)}$ instead to be the vertex such that the RSND demand $(V_{(i,r)}, v_{\ell_1}, 2)$ is not satisfied in $H_i$, then the fault sets are $F = \{e\}$ (if $|V_{(i,\ell)}|=1$) and $F = \{e,f\}$ (if $|V_{(i,\ell)}|=2$), where $e$ is an edge such that $v_{\ell_1}$ and $V_{(i,r)}$ are connected in $G_i \setminus \{e\}$ but not in $H_i \setminus \{e\}$, and $f$ is the edge in $S_{i-1}$ incident on $v_{\ell_2}$.}, 
let $v_{r_1} \in V_{(i,r)}$ be a vertex such that the RSND demand $(V_{(i,\ell)}, v_{r_1}, 2)$ is not satisfied in $H_i$. That is, there exists some edge fault $e \in E$ such that there is a path from $V_{(i,\ell)}$ to $v_{r_1}$ in $G_i \setminus \{e\}$, but there is no path in $H_i \setminus \{e\}$. We will now show that there exists a fault set $F$ with $|F| \leq 2$ such that $s$ and $t$ are connected in $G \setminus F$ but are disconnected in $H \setminus F$. This would imply that $H$ is not feasible, giving a contradiction. There are two cases:
\begin{itemize}
    \item $|V_{(i,r)}| = 1$. Let $F = \{e\}$. There is no path from $V_{(i,\ell)}$ to $V_{(i,r)}$ in $H_i \setminus F$, and therefore no path from $s$ to $t$ in $H \setminus F$. There are no faults between $G_0$ and $G_{i-1}$, inclusive, so by Lemma \ref{lem:no_faults_G} there is a path from $V_{(0,\ell)} = s$ to each vertex in $V_{(i-1,r)}$ in $G \setminus F$, using only edges between $G_0$ and $G_{i-1}$. There are no faults in $S_{i-1}$, so there is also a path from $s$ to each vertex in $V_{(i,\ell)}$, using only edges between $G_0$ and $G_{i}$. As previously stated, there is a path from at least one vertex in $V_{(i,\ell)}$ to $v_{r_1} \in V_{(i,r)}$ in $G_i \setminus F$. There are no faults in $S_{i}$, so there is a path from $v_{r_1}$ to each vertex in $V_{(i+1,\ell)}$, using only edges in $S_{i}$. Finally, there are no faults between $G_{i+1}$ and $G_{p}$ inclusive, so by Lemma \ref{lem:no_faults_G} there is a path from $V_{(i+1,\ell)}$ to $V_{(p,r)} = t$ in $G \setminus F$, using only edges between $G_{i+1}$ and $G_p$. Putting everything together, there is an $s-t$ path in $G \setminus F$, but there is no $s-t$ path in $H \setminus F$, giving a contradiction to the assumption that $H$ is a feasible solution. 
    
    \item $|V_{(i,r)}| = 2$, with $V_{(i,r)} = \{v_{r_1},v_{r_2}\}$. Let $f$ denote the edge in $S_{i}$ incident on $v_{r_2}$, and set $F = \{e, f\}$. In $H \setminus F$, there is no path from $V_{(i,\ell)}$ to $V_{(i+1,\ell)}$. Therefore there is no path from $s$ to $t$ in $H \setminus F$. There are no faults between $G_0$ and $G_{i-1}$, inclusive, and no faults in $S_{i-1}$, so by Lemma \ref{lem:no_faults_G} there is a path from $s$ to each vertex in $V_{(i,\ell)}$ in $G \setminus F$, using only edges between $G_0$ and $G_i$. As previously stated, there is a path from at least one vertex in $V_{(i,\ell)}$ to $v_{r_1} \in V_{(i,r)}$ in $G_i \setminus F$. Since the single edge fault in $S_i$, $f$, is not incident on $v_{r_1}$, there is a path from $v_{r_1}$ to a vertex in $V_{(i+1,\ell)}$, only using the remaining edge in $S_i$. Finally, by Lemma \ref{lem:no_faults_G}, there is a path from each vertex in $V_{(i+1,\ell)}$ to $t$ in $G \setminus F$, which only use edges between $G_{i+1}$ and $G_p$. Putting everything together, there is an $s-t$ path in $G \setminus F$, but there is no $s-t$ path in $H \setminus F$, giving a contradiction to the assumption that $H$ is a feasible solution.
\end{itemize}
Suppose again that $H$ is a feasible, but now suppose for the sake of contradiction that Property 3 in the statement of Lemma \ref{lem:partition_characterize} is not satisfied in $H_i$. Let $v_{\ell_1} \in V_{(i,\ell)}$ and $v_{r_1} \in V_{(i,r)}$ be vertices such that the RSND demand $(v_{\ell_1}, v_{r_1}, 1)$ is not satisfied in $H_i$. This means there is a path from $v_{\ell_1}$ to $v_{r_1}$ in $G_i$, but there is no path in $H_i$. We will now show that there exists a fault set, $F$, with $|F| \leq 2$, such that $s$ and $t$ are connected in $G \setminus F$ but are disconnected in $H \setminus F$. This would mean that $H$ is not a feasible, giving a contradiction. There are three cases:
\begin{itemize}
    \item $|V_{(i,\ell)}|=|V_{(i,r)}| = 1$. Let $F$ be the empty set. There is no path from $V_{(i,\ell)}$ to $V_{(i,r)}$ in $H_i$, and therefore no $s-t$ path in $H$. $G$ is 2-edge connected, so there is a path from $s$ to $t$ in $G$. This is a contradiction to the assumption that $H$ is a feasible.
    
    \item Exactly one of $V_{(i,\ell)}$ and $V_{(i,r)}$ has size 2. Without loss of generality (since the other case is symmetric), let $V_{(i,r)} = \{v_{r_1},v_{r_2}\}$. Let $f$ denote the edge in $S_{i}$ incident on $v_{r_2}$, and set $F = \{f\}$. In $H \setminus F$, there is no path from $V_{(i,\ell)}$ to $V_{(i+1,\ell)}$. Therefore there is no $s-t$ path in $H \setminus F$. There are no faults between $G_0$ and $G_{i}$, inclusive, so by Lemma \ref{lem:no_faults_G} there is a path from $s$ to each vertex in $V_{(i,r)}$ in $G \setminus F$, using only edges between $G_0$ and $G_i$. There is also a path from $v_{r_1}$ to a vertex in $V_{(i+1,\ell)}$, only using an edge in $S_{i}$. Finally, by Lemma \ref{lem:no_faults_G}, there is a path from each vertex in $V_{(i+1,\ell)}$ to $t$, using only edges between $G_{i+1}$ and $G_p$. Putting everything together, there is an $s-t$ path in $G \setminus F$, but there is no $s-t$ path in $H \setminus F$, giving a contradiction to the assumption that $H$ is a feasible solution.
    
    \item $V_{(i,\ell)}$ and $V_{(i,r)}$ have size 2. Let $V_{(i,\ell)} = \{v_{\ell_1},v_{\ell_2}\}$ and $V_{(i,r)} = \{v_{r_1},v_{r_2}\}$. We also let $f_1$ denote the edge in $S_{i-1}$ incident on $v_{\ell_2}$, and let $f_2$ denote the edge in $S_{i}$ incident on $v_{r_2}$. Set $F = \{f_1, f_2\}$. In $H \setminus F$, there is no path from $V_{(i-1,r)}$ to $V_{(i+1,\ell)}$. Therefore, there no $s-t$ path in $H \setminus F$. There are no faults between $G_0$ and $G_{i-1}$, inclusive, so by Lemma \ref{lem:no_faults_G} there is a path from $s$ to each vertex in $V_{(i-1,r)}$ in $G \setminus F$, using only edges between $G_0$ and $G_{i-1}$. There is only one fault in $S_{i-1}$, so there is a path from one vertex in $V_{(i-1,r)}$ to $v_{\ell_1} \in V_{(i,\ell)}$, only using the remaining edge in $S_{i-1}$. There is also a path in $G_i \setminus F$ from $v_{\ell_1}$ to $v_{r_1}$, by our assumption. Finally, $v_{r_1}$ has a path to a vertex in $V_{(i+1,\ell)}$, which only uses the remaining edge in $S_i$. By Lemma \ref{lem:no_faults_G}, each vertex in $V_{(i+1,\ell)}$ has a path to $t$, using only edges between $G_{i+1}$ and $G_p$. Putting everything together, there is an $s-t$ path in $G \setminus F$, but there is no $s-t$ path in $H \setminus F$, giving a contradiction to the assumption that $H$ is a feasible solution.
\end{itemize}

\subsection{If}
Now we prove that the properties stated in Lemma~\ref{lem:partition_characterize} are sufficient. Suppose all edges in $\cup_{j=0}^{p-1} S_j$ are in $H$, and suppose all 3 properties in the statement of Lemma \ref{lem:partition_characterize} are met for all $H_i$. For all possible fault sets $F$, with $|F| \leq 2$, we will show that if $s$ and $t$ are connected in $G \setminus F$, they must also be connected in $H \setminus F$, and therefore $H$ is feasible. Let $F = \{f_1, f_2\}$ be the fault set, and suppose $s$ and $t$ are connected in $G \setminus F$. We will first show that if there are no faults between $H_i$ and $H_j$, then in $H$, each vertex in $V_{(j,r)}$ is reachable from $V_{(i,\ell)}$, and each vertex in $V_{(i,\ell)}$ has a path to $V_{(j,r)}$. 

\begin{lemma}
\label{lem:no_faults_H}
Let $H$ be a subgraph of $G$, and suppose that all properties in the statement of Lemma \ref{lem:partition_characterize} are met by $H_i$ for all $i$. Consider the subchain that starts at $H_i$ and ends at $H_j$, where $i \leq j$. Then, there is a path from each vertex in $V_{(i,\ell)}$ to the vertex set $V_{(j,r)}$, and a path from the vertex set $V_{(i,\ell)}$ to each vertex in $V_{(j,r)}$ in $H$. These paths only use edges within the subchain that starts at $H_i$ and ends at $H_j$.
\end{lemma}
\begin{proof}
We have shown in Lemma \ref{lem:no_faults_G} that if subgraph $G_k$ is in the subchain that begins at $G_i$ and ends at $G_j$, then each vertex in $V_{(k,r)}$ is reachable from $V_{(k,\ell)}$, and $V_{(k,r)}$ is reachable from each vertex in $V_{(k,\ell)}$. For each subgraph $H_k$ in the subchain from $H_i$ to $H_j$, the RSND demands $\{(u,v,1) : (u,v) \in V_{(i,\ell)} \times V_{(i,r)}\}$ are satisfied (Property 3 of Lemma \ref{lem:partition_characterize}). That is, if $v_{k\ell} \in V_{(k,\ell)}$ and $v_{kr} \in V_{(k,r)}$ are connected in $G$, then they are connected in $H$. Therefore, we also have that in $H_k$, each vertex in $V_{(k,r)}$ is reachable from $V_{(k,\ell)}$, and $V_{(k,r)}$ is reachable from each vertex in $V_{(k,\ell)}$. This also implies that in $H$, each vertex in $V_{(j,r)}$ is reachable from $V_{(i,\ell)}$, and each vertex in $V_{(i,\ell)}$ has a path to $V_{(j,r)}$, using only edges within the subchain from $H_i$ to $H_j$. This can been seen via a proof by induction on the number of components into the chain.
\end{proof}

For all $i$, let component $R_i$ with separator $S_{i-1}$ (if it exists) together be the \textit{$i$th section} of the chain. Section $i$ is considered earlier in the chain than section $j$ if $i < j$. We divide the $s-t$ 2-chain into subchains as follows. Suppose edge fault $f_1$ is in section $i$, and $f_2$ is in section $j$, where $i \leq j$. Then the first subchain, $L_1$, begins at $R_0$ and ends at $R_i$, inclusive. The second subchain, $L_2$, begins at $R_{0}$ and ends at $R_k$, inclusive. If $i = k$, then $L_1$ and $L_2$ are the same subchain. Let $L_\alpha^G$ denote the subchain of $G$ induced by $L_\alpha$, and let $L_\alpha^H$ denote the subchain of $H$ induced by $L_\alpha$. 

We now prove a series of lemmas. Lemma \ref{lem:1_fault} states which vertices at the end of $L_1$ are reachable in $L_1^H \setminus F$ given what is reachable in $L_1^G \setminus F$. Lemma \ref{lem:2nd_1_fault} uses Lemma \ref{lem:1_fault} to prove that when the two edge faults are in different sections of the chain and $s$ and $t$ are connected in $G \setminus F$, then $s$ and $t$ are connected in $H \setminus F$.  Lemma \ref{lem:2_fault} shows that when both edge faults are in the same section, $s$ and $t$ are connected in $H \setminus F$ if they are connected in $G \setminus F$

\begin{lemma}
\label{lem:1_fault}
Let $F = \{f_1, f_2\}$ be the fault set, where $f_1$ and $f_2$ are in sections $i$ and $j$, respectively, with $i < j$. Suppose there is an $s-t$ path in $G \setminus F$. Consider subchain $L_1$, as defined above. Then, every vertex in $V_{(i,r)}$ that is reachable from $s$ in $L_1^G \setminus F$ is also reachable from $s$ in $L_1^H \setminus F$.
\end{lemma}

\begin{proof}
There are no faults from $R_0$ to $R_{i-1}$, so by Lemmas \ref{lem:no_faults_G} and \ref{lem:no_faults_H}, in $G$ and in $H$, there is a path from $s$ to each vertex in $V_{(i-1,r)}$, using only edges between $G_0$ and $G_{i-1}$, and between $H_0$ and $H_{i-1}$, respectively. There are two cases: $f_1 \in R_i$ and $f_1 \in S_{i-1}$. 

\begin{itemize}
    \item Suppose first that $f_1 \in R_i$. There is a path from $s$ to $t$ in $G \setminus F$, so there must also be a path from $V_{(i,\ell)}$ to at least one vertex in $V_{(i,r)}$ in $G_i \setminus F$. Suppose there is a path in $G_i \setminus F$ from $V_{(i,\ell)}$ to $v_{ir} \in V_{(i,r)}$; that is, there is a path from $V_{(i,\ell)}$ to $v_{ir}$ after the removal of $f_1$. Property 2 from Lemma \ref{lem:partition_characterize} is satisfied in $H_i$. Therefore, in $H_i$, $V_{(i,\ell)}$ and $v_{ir}$ must also be connected after the removal of $f_1$. Thus, if a vertex in $V_{(i,r)}$ is reachable from $V_{(i,\ell)}$ in $G_i \setminus F$, then that same vertex is reachable from $V_{(i,\ell)}$ in $H_i \setminus F$. Note that since there are no faults in $S_{i-1}$, we can also say that if a vertex $u \in V_{(i,r)}$ is reachable from the vertex set $V_{(i-1,r)}$ in $L^G_1 \setminus F$, then $u$ is reachable from $V_{(i-1,r)}$ in $L^H_1 \setminus F$. 
    
    \item Now suppose that $f_1 \in S_{i-1}$. Without loss of generality, let $f_1$ be incident on vertex $v_{i\ell_1} \in V_{(i,\ell)}$. 
    If $|V_{(i,\ell)}| = 1$, then in $G$ and in $H$, $v_{i\ell_1}$ is adjacent to both vertices in $V_{(i-1,r)}$. Therefore, after the removal of $f_1$, $v_{i\ell_1}$ is still adjacent to a vertex in $V_{(i-1,r)}$ in $G \setminus F$ and in $H \setminus F$. Additionally, there are no faults in $R_i$, so by Lemmas \ref{lem:no_faults_G} and \ref{lem:no_faults_H}, there is a path from $V_{(i,\ell)} = \{v_{i\ell_1}\}$ to each vertex in $V_{(i,r)}$ in $G_i \setminus F$ and in $H_i \setminus F$. Putting it all together, in $L^G_1 \setminus F$ and in $L^H_1 \setminus F$, there is a path from $V_{(i-1,r)}$ to each vertex in $V_{(i,r)}$ that uses only edges in $S_{i-1}$ and in $R_i$.
    If $|V_{(i,\ell)}| = 2$, then let $V_{(i,\ell)} = \{v_{i\ell_1}, v_{i\ell_2} \} $. Recall that $f_1 \in S_{i-1}$ is incident on $v_{i\ell_1} \in V_{(i,\ell)}$. We therefore have that any path into $V_{(i,\ell)}$ from $V_{(i-1,r)}$ must visit $v_{i\ell_2}$. Since there is an $s-t$ path in $G \setminus F$, there must also be a path from $v_{i\ell_2}$ to $V_{(i,r)}$ in $G_i \setminus F$. Property 3 from Lemma \ref{lem:partition_characterize} is satisfied in $H_i$. Therefore, if there is a path from $v_{i\ell_2}$ to a vertex $v_{ir_1} \in V_{(i,r)}$ in $G_i$, then there is a path from $v_{i\ell_2}$ to $v_{ir_1}$ in $H_i$. Putting everything together, we have the following: If there is a path from $V_{(i-1,r)}$ to a vertex $u \in V_{(i,r)}$ in $L^G_1 \setminus F$, then there is a path from $V_{(i-1,r)}$ to $u$ in $L^H_1 \setminus F$. 
\end{itemize}

We have shown that if section $i$ has exactly one fault, then every vertex in $V_{(i,r)}$ that is reachable from $V_{(i-1,r)}$ in $L_1^G \setminus F$ is also reachable from $V_{(i-1,r)}$ in $L_1^H \setminus F$. Recall that in $G$ and in $H$, there is a path from $s$ to each vertex in $V_{(i-1,r)}$ that only uses edges between $G_0$ and $G_{i-1}$, and between $H_0$ and $H_{i-1}$, respectively. We therefore have that every vertex in $V_{(i,r)}$ that is reachable from $s$ in $L^G_1 \setminus F$ is also reachable from $s$ in $L^H_1 \setminus F$.
\end{proof}

In the following lemma, we will use Lemma \ref{lem:1_fault} to prove that if the edge faults in $F$ are in different sections of the $s-t$ 2-chain, then there is an $s-t$ path in $G \setminus F$ if and only if there is an $s-t$ path in $H \setminus F$.

\begin{lemma}
\label{lem:2nd_1_fault}
 Let $F = \{f_1, f_2\}$ be the fault set, where $f_1$ and $f_2$ are in sections $i$ and $j$, respectively, with $i<j$. Suppose there is an $s-t$ path in $G \setminus F$. Consider subchain $L_2$, as defined above. Then, at least one vertex in $V_{(j,r)}$ is reachable from $s$ in $L^H_2 \setminus F$, and there is an $s-t$ path in $H \setminus F$.
\end{lemma}

\begin{proof}
There is an $s-t$ path in $G \setminus F$, so at least one vertex in $V_{(i,r)}$ must be reachable from $s$ in $L^G_1 \setminus F$. Let $v_{ir_1}$ be this vertex. By Lemma \ref{lem:1_fault}, we also have that $v_{ir_1}$ is reachable from $s$ in $L^H_1 \setminus F$.  There are no faults between sections $i$ and $j$, \textit{not} inclusive, so using Lemmas \ref{lem:no_faults_G} and \ref{lem:no_faults_H}, we can say that in both $G \setminus F$ and in $H \setminus F$, there is a path from $v_{ir_1}$ to $V_{(j-1, r)}$, using only edges in $S_i$ and in the subchain from $G_{i+1}$ to $G_{j-1}$, or in the subchain from $H_{i+1}$ to $H_{j-1}$, respectively. Additionally, Property 3 of Lemma \ref{lem:partition_characterize} is met for all $H_i$. Therefore, if $G$ has a path from $v_{ir_1}$ to a particular vertex $u \in V_{(j-1, r)}$ that only uses edges in $S_i$ and in the subchain from $G_{i+1}$ to $G_{j-1}$, then $H$ also has a path from $v_{ir_1}$ to $u$ that only uses edges in $S_i$ and in the subchain from $H_{i+1}$ to $H_{j-1}$. We therefore have that every vertex in $V_{(j-1, r)}$ that is reachable from $s$ using only edges between $G_0$ and $G_{j-1}$ is also reachable from $s$ using only edges between $H_0$ and $H_{j-1}$. Now, we have two cases: $f_2$ is in $R_j$ or $f_2$ is in $S_{j-1}$.
\begin{itemize}
    \item We first consider the case with $f_2 \in R_j$. There is an $s-t$ path in $G \setminus F$, and let $P$ be such a path. There must be at least one vertex in $V_{(j-1, r)}$ that is in $P$ in $G \setminus F$ (otherwise, $P$ would not be an $s-t$ path in $G \setminus F$). Let $v_{j-1,r}$ be such a vertex. As proved in the first paragraph of this proof, we also have that there is a path from $s$ to $v_{j-1,r}$ in $H \setminus F$, using only edges between $H_0$ and $H_{j-1}$. Let $v_{j\ell}$ be the vertex in $V_{(j, \ell)}$ that is in $P$ and adjacent to $v_{j-1,r}$. There are no faults in $S_{j-1}$, so there is also a path from $s$ to $v_{j\ell}$ that only uses edges between $G_0$ and $G_{j-1}$, or between $H_0$ and $H_{j-1}$, and an edge in $S_{j-1}$. Since $P$ is an $s-t$ path in $G \setminus F$ that uses $v_{j\ell}$, there is a path from $v_{j\ell}$ to $t$ in $G \setminus F$ that only uses edges between $G_{j}$ to $G_{p}$. This also means there is a path from $v_{j\ell}$ to $V_{(j,r)}$ in $G_j \setminus F$. Since Property 2 from Lemma \ref{lem:partition_characterize} is met on subgraph $H_j$, there must also be a path from $v_{j\ell}$ to $V_{(j,r)}$ in $H_j \setminus F$. There are no faults in $S_j$, implying that in $H \setminus F$, there is also a path from $v_{j\ell}$ to a vertex in $V_{(j+1,\ell)}$, using only edges in $H_j$ and an edge in $S_j$.
    
    \item We now consider the case with $f_2 \in S_{j-1}$. At least one vertex in $V_{(j-1,r)}$ is reachable from $s$ in $G \setminus F$, using only edges between $G_0$ and $G_{j-1}$ (otherwise there would be no $s-t$ path in $G \setminus F$). Suppose first that all vertices in $V_{(j-1,r)}$ are reachable from $s$ in $G \setminus F$, using only edges between $G_0$ and $G_{j-1}$. Then, each vertex in $V_{(j-1,r)}$ is reachable from $s$ in $H \setminus F$ as well, using only edges between $H_0$ and $H_{j-1}$ (proved in paragraph 1 of this proof). There is one remaining edge in $S_{j-1}$, so in $G \setminus F$ and in $H \setminus F$, there must be a path from $V_{(j-1,r)}$ to one vertex $v_{j\ell} \in V_{(j, \ell)}$ that only uses the remaining edge in $S_{j-1}$. 
    Now suppose exactly one vertex in $V_{(j-1,r)}$ is reachable from $s$ in $G \setminus F$, using only edges between $G_0$ and $G_{j-1}$.  Let $v_{j-1,r}$ be this vertex. Therefore (proved in paragraph 1 of this proof), $v_{j-1,r}$ is also reachable from $s$ in $H \setminus F$, using only edges between $H_0$ and $H_{j-1}$. We can assume that $|V_{(j-1,r)}| = 2$, since the $|V_{(j-1,r)}| = 1$ case is covered by the previous argument. If $f_2$ is the edge in $S_{j-1}$ that is incident on $v_{j-1,r}$, then there is no $s-t$ path in $G \setminus F$. This contradicts the assumption that there is an $s-t$ path in $G \setminus F$. Therefore, there must be a path in $G \setminus F$, and in $H \setminus F$, from $v_{j-1,r}$ to a vertex $v_{j\ell} \in V_{(j, \ell)}$, using only the non-fault edge in $S_{j-1}$.  
\end{itemize}
We have shown that in $H \setminus F$, there is a path from $s$ to at least one vertex $v_{j\ell}$ in $V_{(j, \ell)}$ that only uses edges between $H_0$ and $H_{j-1}$ and in $S_{j-1}$. Additionally, there are no faults between $H_{j}$ and $H_p$, so by Lemma \ref{lem:no_faults_H}, $v_{j\ell}$ has a path to $V_{(p,r)} = \{t\}$ that only uses edges between $H_{j}$ and $H_p$. Putting it all together, in $H \setminus F$, there is a path from $s$ to $t$. \qedhere
\end{proof}

Now we will show that if the edge faults in $F$ are in the same section of the $s-t$ 2-chain, then there is an $s-t$ path in $G \setminus F$ if and only if there is an $s-t$ path in $H \setminus F$

\begin{lemma}
\label{lem:2_fault}
Let $F = \{f_1, f_2\}$ be the fault set, where $f_1$ and $f_2$ are both in section $i$. Suppose there is an $s-t$ path in $G \setminus F$. Consider subchain $L_1$, as defined above. Then, at least one vertex in $V_{(i,r)}$ is reachable from $s$ in $L^H_1 \setminus F$, and there is an $s-t$ path in $H \setminus F$.
\end{lemma}
\begin{proof}
There are no faults between $R_0$ and $R_{i-1}$, so by Lemmas \ref{lem:no_faults_G} and \ref{lem:no_faults_H}, in $G \setminus F$ and in $H \setminus F$, each vertex in $V_{(i-1,r)}$ is reachable from $s$, using only edges between $G_0$ and $G_{i-1}$, or between $H_0$ and $H_{i-1}$, respectively. We have two cases: Either both $f_1$ and $f_2$ are in $R_i$ or, without loss of generality, $f_1 \in S_{i-1}$ and $f_2 \in R_i$. Note that $f_1$ and $f_2$ cannot both belong in $S_{i-1}$ because this would contradict the assumption that there is an $s-t$ path in $G \setminus F$. 
\begin{itemize}
    \item First consider the case with $f_1$ and $f_2$ in $R_i$. In $G_i$, there are at least 3 edge-disjoint paths from $V_{(i,\ell)}$ to $V_{(i,r)}$; otherwise, by Menger's theorem, there is a cut of size at most two that separates $V_{(i,\ell)}$ from $V_{(i,r)}$. This would mean that $S_i$ cannot be an important $(V_{(i,\ell)},t)$-separator of $G \setminus \cup_{j=0}^{i-1} R_j $. Since Property 1 in Lemma \ref{lem:partition_characterize} is met in $H_i$, there must also be 3 or more edge-disjoint paths from $V_{(i,\ell)}$ to $V_{(i,r)}$ in $H_i$. Therefore, there is a path from $V_{(i,\ell)}$ to $V_{(i,r)}$ in $H_i \setminus F$. There are no faults in $S_{i-1}$, so we can also say there is a path from $V_{(i-1,r)}$ to $V_{(i,r)}$ in $L_1^H \setminus F$. 
    
    \item Next, consider the case with $f_1 \in S_{i-1}$ and $f_2 \in R_i$. Suppose without loss of generality that $f_1$ is incident on vertex $v_{i\ell_1} \in V_{(i, \ell)}$. We first consider the case with $|V_{(i, \ell)}| = 1$. Since $v_{i\ell_1}$ is adjacent to both vertices in $V_{(i-1, r)}$, there is still a path from $V_{(i-1,r)}$ to $v_{i\ell_1}$ in $H\setminus F$ using only the remaining edge in $S_{i-1}$. Additionally, there are at least 3 edge-disjoint paths from $v_{i\ell_1}$ to $V_{(i,r)}$ in $H_i$. There is only one fault, $f_2$, in $R_i$. Thus, in $H_i \setminus F$, there must be a path from $v_{i\ell_1}$ to $V_{(i,r)}$. 
    Now we consider the case with $|V_{(i, \ell)}| = 2$. Let $V_{(i, \ell)} = \{ v_{i\ell_1}, v_{i\ell_2} \} $. Since $f_1 \in S_{i-1}$ is incident on $v_{i\ell_1}$, any path from $V_{(i-1,r)}$ to $V_{(i, \ell)}$ in $G \setminus F$ and in $H \setminus F$ must use the edge incident on $v_{i\ell_2}$, and must visit $v_{i\ell_2}$. Since there is a path from $s$ to $t$ in $G \setminus F$, there must also be a path from $v_{i\ell_2}$ to $V_{(i,r)}$ in $G_i \setminus F$. Property 2 in Lemma \ref{lem:partition_characterize} is satisfied in $H_i$. Therefore, if there is a path from $v_{i\ell_2}$ to $V_{(i,r)}$ in $G_i \setminus F$, then there must also be a path from $v_{i\ell_2}$ to $V_{(i,r)}$ in $H_i \setminus F$. Therefore, we have a path from $V_{(i,\ell)}$ (and from $V_{(i-1,r)}$)) to $V_{(i,r)}$ in $H \setminus F$.
\end{itemize}
We have shown that there is a path from $V_{(i-1,r)}$ to $V_{(i,r)}$ in $L_1^H \setminus F$. Additionally, there are no faults in $S_i$, so in $H \setminus F$ there is also a path from each vertex in $V_{(i,r)}$ to $V_{(i+1,\ell)}$ that only uses an edge in $S_i$. Finally, there are no faults between $H_{i+1}$ and $H_p$, inclusive, so by Lemma \ref{lem:no_faults_H}, each vertex in $V_{(i+1,\ell)}$ has a path to $V_{(p,r)} = \{t\}$, using only edges between $H_{i+1}$ and $H_p$. Putting everything together, in $H \setminus F$, there is a path from $s$ to $t$.
\end{proof}

Lemmas~\ref{lem:2nd_1_fault} and \ref{lem:2_fault} together clearly imply the ``if'' direction of Lemma~\ref{lem:partition_characterize}.

\section{Proofs from Section~\ref{sec:3RSND-alg}} \label{app:3RSND-alg}
\begin{proof}[Proof of Lemma~\ref{lem:3RSND-feasible}]
For each $i$, let $H_{i}$ denote the subgraph of $H$ induced by $R_i$ and let $G_i$ denote the subgraph of $G$ induced by $R_i$.  We will show that $H$ satisfies the conditions of Lemma~\ref{lem:partition_characterize}, and hence is feasible.  By construction, $H$ contains all edges $S$ in the important separators.  

To show property 1 of Lemma~\ref{lem:partition_characterize}, recall that in each $H_i$ we included the edges selected via a min-cost flow algorithm from $V_{(i,\ell)}$ to $V_{(i,r)}$ with flow $3$.  Since there are at least three edge-disjoint paths from $V_{(i,\ell)}$ to $V_{(i,r)}$ in $G_i$ (by Lemma~\ref{lem:partition_characterize} since $G$ itself is feasible), this will return three edge-disjoint paths from $V_{(i,\ell)}$ to $V_{(i,r)}$.  Hence $H$ satisfies the first property.

Property 2 of Lemma~\ref{lem:partition_characterize} is direct from the algorithm, since $H_i$ includes the output of the 2-RSND algorithm from Theorem~\ref{thm:2RSND} when run on demands $\left\{(V_{(i,\ell)},v_r,2) : v_r \in V_{(i,r)}\right\} \cup \left\{(V_{(i,r)},v_\ell,2) : v_\ell \in V_{(i,\ell)}\right\}$.  Similarly, within each component $H_{i}$ in the $s-t$ 2-chain, the edges selected by the Steiner Forest algorithm form a path from vertex $v_{\ell} \in V_{(i,\ell)}$ to vertex $v_r \in V_{(i,r)}$ if $v_{\ell}$ and $v_r$ are connected in $G$. This satisfies Property 3 in Lemma \ref{lem:partition_characterize}. 
\end{proof}

\begin{proof}[Proof of Lemma~\ref{lem:3RSND-cost}]
Let $H_{i} = H[R_i]$ be the subgraph of $H$ induced by $R_i$, and let $H_{i}^* = H^*[R_i]$ be the subgraph of the optimal solution induced by $R_i$. We also let $H_{i}^M$ denote the subgraph of $H_i$ returned by the min-cost flow algorithm run on $R_i$ (i.e., the set of edges with non-zero flow), let $H_{i}^{N^1}$ and $H_{i}^{N^2}$ denote the subgraphs returned by the first and second 2-approximation 2-RSND algorithms run on $R_i$, respectively, and we let $H_{i}^{F}$ denote the subgraph of $H_i$ returned by the Steiner Forest algorithm on $R_i$. We also let $M_{i}^{*}$ be the optimal solution to the Minimum-Cost Flow instance on $R_i$, let $N_{i}^{1^*}$ and $N_{i}^{2^*}$ be the optimal solutions to the first and second 2-RSND instances on $R_i$, respectively, and let $F_{i}^{*}$ be the optimal solution to the Steiner Forest instance on $R_i$. Subgraph $H_{i}^M$ is given by an exact algorithm, subgraphs $H_{i}^{N^1}$ and $H_{i}^{N^2}$ are given by a 2-approximation algorithm, and subgraph $H_{i}^{F}$ is given by a $\left(2-\frac{1}{k}\right)$-approximation algorithm. Note that there are at most $4$ terminal pairs in the Steiner Forest instance, so $k \leq 4$ and the algorithm gives a $\frac{7}{4}$-approximation. Hence we have the following for each component $R_i$:
\begin{align*}
    w(H_{i}^M) &= w(M_{i}^{*}) &  w(H_{i}^{N^1}) &\leq 2 w(N_{i}^{1^*}) \\
    w(H_{i}^{N^2}) &\leq 2 w(N_{i}^{2^*}) & w(H_{i}^{F}) &\leq \frac{7}{4} w(F_{i}^{*}). 
\end{align*}
Summing over all components in the chain, we get the following:
\begin{align*}
    \sum_{i=0}^{p} w(H_{i}^M) &= \sum_{i=0}^{p} w(M_{i}^{*}) &
    \sum_{i=0}^{p} w(H_{i}^{N^1}) &\leq 2 \cdot \sum_{i=0}^{p} w(N_{i}^{1^*}) \\
    \sum_{i=0}^{p} w(H_{i}^{N^2}) &\leq 2 \cdot \sum_{i=0}^{p} w(N_{i}^{2^*}) &
    \sum_{i=0}^{p} w(H_{i}^{F}) &\leq \frac{7}{4}  \cdot \sum_{i=0}^{p} w(F_{i}^{*}). 
\end{align*}
We also have that
\begin{align*}
    w(H_{i}) &\leq w(H_{i}^M) + w(H_{i}^{N^1}) + w(H_{i}^{N^2})+ w(H_{i}^{F}).
\end{align*}
Summing over all components in the chain and then substituting the above, we get the following:
\begin{align*}
    \sum_{i=0}^{p} w(H_{i}) &\leq \sum_{i=0}^{p} w(H_{i}^M) + \sum_{i=0}^{p} w(H_{i}^{N^1}) + \sum_{i=0}^{p} w(H_{i}^{N^2}) + \sum_{i=0}^{p} w(H_{i}^{F}) \\
    &\leq \sum_{i=0}^{p} w(M_{i}^{*}) + 2 \cdot \sum_{i=0}^{p} w(N_{i}^{1^*}) + 2 \cdot \sum_{i=0}^{p} w(N_{i}^{2^*}) + \frac{7}{4}  \cdot \sum_{i=0}^{p} w(F_{i}^{*}). 
\end{align*}
The optimal subgraph $H^*$ is a feasible solution, so by Lemma \ref{lem:partition_characterize}, each property in the lemma statement must be met on subgraph $H^*_i$ for all $i$. For all properties in the lemma to be satisfied on $H^*_i$, the set of edges $E(H^*_i)$ must be a feasible solution to each of the Minimum-Cost Flow, 2-RSND, and Steiner Forest instances on $R_i$. Therefore, the cost of $H^*_i$ must be at least the cost of the optimal solution to each of the Minimum-Cost Flow, 2-RSND, and Steiner Forest instances. We therefore have the following:
\begin{align*}
    \sum_{i=0}^{p} w(H_{i}) &\leq \sum_{i=0}^{p} w(H^*_{i}) + 2 \cdot \sum_{i=0}^{p} w(H^*_{i}) + 2 \cdot \sum_{i=0}^{p} w(H^*_{i}) + \frac{7}{4}  \cdot \sum_{i=0}^{p} w(H^*_{i}) \leq \frac{27}{4} \cdot \sum_{i=0}^{p} w(H_{i}^*).
\end{align*}
Finally, we must account for the edges between components in the $s-t$ 2-chain. Let $S$ be the set of edges between components in the chain that are included in the algorithm solution, and let $S^*$ be the set of edges between components included in the optimal solution.  By Lemma~\ref{lem:partition_characterize}, any feasible solution must include all edges between the components of the chain. We therefore have that $S = S^*$ and we get the following:
\begin{align*}
    w(H) = \sum_{i=0}^{p} w(H_{i}) + w(S) &\leq \frac{27}{4}  \sum_{i=0}^{p} w(H_{i}^*) + w(S) \leq \frac{27}{4} \left( \sum_{i=0}^{p} w(H_{i}^*) + w(S^*)\right) \leq \frac{27}{4} w(H^*). \qedhere
\end{align*}
\end{proof}

\fi

\end{document}